\documentclass[11pt]{article}
\UseRawInputEncoding
\usepackage[paper=letterpaper, margin=1in]{geometry}

\usepackage[paper=letterpaper, margin=1in]{geometry}
\usepackage{lineno}

\usepackage{amsmath, amsthm,dsfont}
\usepackage{amssymb}
\usepackage{nicefrac}
\usepackage[noblocks]{authblk}
\usepackage[export]{adjustbox}
\input{macros.tex}
\usepackage[dvipsnames]{xcolor}
\usepackage{thmtools}
\usepackage{thm-restate}
\usepackage{cleveref}
\usepackage{subcaption}

\usepackage{soul}

\newcommand{\E}{\mathbb{E}}
\newcommand{\T}{\mathbb{T}}
\renewcommand{\R}{\mathbb{R}}
\newcommand{\norm}[1]{\left\lVert#1\right\rVert}
\newcommand{\Var}{\mathrm{Var}}
\newcommand{\lmax}{\ell_{max}}
\newcommand{\Over}{\mathtt{Over}}

\newcommand{\Z}{\mathbb{Z}}

\newcommand{\eps}{\varepsilon}
\renewcommand{\epsilon}{\eps}

\newcommand{\opt}{\mathtt{opt}}
\newcommand{\cauchy}{\mathtt{cauchy}}

\title{Intermittent Inverse-Square L\'evy Walks are Optimal for Finding Targets of All Sizes}

\author{
  Brieuc Guinard 
  and
  Amos Korman 
\\
IRIF, CNRS and University~of Paris, France}

\begin{document}
\date{}
\maketitle
\thispagestyle{empty}

\begin{abstract}

{\bf L\'evy walks are random walk processes whose step-lengths follow a long-tailed power-law distribution. Due to their abundance as movement patterns of biological organisms, significant theoretical efforts have been devoted to identifying the foraging circumstances that would make such patterns advantageous. However, despite extensive research, there is currently no  mathematical proof indicating that L\'evy walks are, in any manner, preferable strategies in higher dimensions than one. Here we prove that in finite two-dimensional terrains, the inverse-square L\'evy walk strategy is extremely efficient at finding sparse targets of arbitrary size and shape. Moreover, this holds even under the weak model of intermittent detection. Conversely, any other intermittent L\'evy walk fails to efficiently find either large targets or small ones. Our results shed new light on the {\em L\'evy foraging hypothesis}, and are thus expected to impact future experiments on animals performing L\'evy walks.}
\end{abstract}

\section*{Introduction}
 L\'evy walks~\cite{Levy-review,viswanathan2008levy,ThePhysicsOfForaging} are super-diffusive random walk processes, characterised by frequent short move-steps and rarer long re-location steps. Their hallmark is a step-length distribution with a heavy power-law tail:
$p(\ell)\sim 1/\ell^\mu$, for some fixed $1<\mu\leq 3$. The efficiency of L\'evy walks as a foraging strategy was first suggested by Shlesinger and Klafter in 1986 \cite{Shlesinger}. An influential breakthrough was later established in 1999 by Viswanathan et al.~\cite{Viswanathan2}, arguing that when food patches are scarce and non-destructive, the L\'evy walk with exponent $\mu=2$,  hereafter termed {\em Cauchy walk}, consumes more food than other  L\'evy walks. This optimality claim  initiated a burst of experimental studies identifying L\'evy-like movement patterns in a myriad of biological systems~\cite{Tcell,ariel2015swarming,WeierstrassianSnails,bees,reynolds2007displaced,deer,MarinePredatorEnvironment,MarinePredator,Viswanathan2,humphries2012foraging,Boyer1743,humans2,HunterGatherer,reynolds2007free,bartumeus2003helical,Jager}, 
including multiple scenarios identifying Cauchy patterns~\cite{reynolds2007displaced,Tcell,MarinePredatorEnvironment,Jager,reynolds2007free,bartumeus2003helical}. 

The aforementioned quest for L\'evy patterns in biology was largely driven by the {\em L\'evy foraging hypothesis}~\cite{viswanathan2008levy},  stating that since Cauchy walks can optimize search efficiencies, then natural selection should have led to the adaptation of Cauchy walks foraging. Despite concerns about susceptibility to model assumptions~\cite{james2011assessing,plank2008optimal}, the optimality claim of Viswanathan et al.~\cite{Viswanathan2} has been the primary theoretical argument for the optimality of Cauchy walks, and has thus served as the basis on which the L\'evy foraging hypothesis was built. However, while this optimality claim is well-founded in one-dimensional topologies \cite{buldyrev2001average}, its validity in higher dimensions has been under debate \cite{3D}.
In particular, according to the recent result by Levernier et al.~\cite{new-paper}, Cauchy walks are not better than other L\'evy walks in the setting of \cite{Viswanathan2}. This controversy  suggests that the justification of the L\'evy foraging hypothesis may rely on different foraging assumptions than the ones in the work of Viswanathan et al.~\cite{Viswanathan2}.

In this context, it is natural to ask the following question:
 which natural conditions would make L\'evy walks, and particularly Cauchy walks, a favorable foraging strategy? Conclusive answers to this question already exist with respect to one-dimensional topologies  \cite{Lomholt,Viswanathan2}.  
For example, Lomholt et al.~\cite{Lomholt} restricted attention to  {\em intermittent} strategies \cite{Review-inter,reynolds2006intermittent}, in which detection is possible only at the short pauses between random steps and not while moving ballistically. By comparing to other intermittent strategies, the authors argued that the intermittent Cauchy walk is an optimal search strategy in finite one-dimensional terrains. Regarding two-dimensional terrains, extensive simulations by Humphries and Sims~\cite{humphries2014optimal} suggested that Cauchy walks are somewhat favorable when foraging under heterogeneous prey distributions. However, until now there has not been any rigorous argument identifying any type of circumstances in two dimensional terrains that make L\'evy walks, of any kind, advantageous. 

In this paper, we prove that in finite two-dimensional domains,  the (truncated) intermittent Cauchy walk is an optimal search strategy when the goal is to quickly find targets of arbitrary sizes. Other L\'evy walks may perform as well as the Cauchy walk, however, to do so they must be tuned to the size of the target. In fact, we prove that every intermittent L\'evy walk other than Cauchy is extremely inefficient with respect to a large range of target sizes. In contrast, and remarkably, the intermittent Cauchy walk stands out as the only intermittent process that is efficient across all target scales without the need for any adaptation.

Robustness to target scales is expected to yield fitness advantages as searching for targets that significantly vary in size is prevalent in biology, including in scenarios where L\'evy patterns have been reported. To name a few examples, this occurs when marine predators search for fish patches~\cite{MarinePredator,MarinePredatorEnvironment}, albatrosses forage on patches of squid and fish~\cite{weimerskirch2007seabirds}, bees search for assemblages of flowers~\cite{bees}, fruit flies explore their landscape \cite{reynolds2007free}, 
marine  dinoflagellate search for patches of  phytoplankton 
\cite{bartumeus2003helical}, swarming bacteria search for food concentrations  \cite{ariel2015swarming}, T-cells search for an invasion of pathogens \cite{Tcell}, and even when the eye scans the visual field \cite{eyes}. 

\section*{Model}
We consider an idealized model in which a searcher aims to quickly find a single target in a finite two-dimensional terrain with periodic boundary conditions, modelled as a square torus $\T_n=[-\sqrt{n}/2,\sqrt{n}/2]^2\subset \R^2$, whose area is $n$. 
Note that this geometry mimics both relevant situations of a single target in a finite domain and of infinitely many regularly spaced targets in an infinite domain, as considered in \cite{Viswanathan2}. Indeed, given a certain density of targets, one can find $n$ and tile the space into squares of area $n$, such that in each square there is approximately one target. Now, moving ballistically from one square to an adjacent square can be viewed as moving on the torus with periodic boundaries. Of course, the target in one square is not necessarily located in the same position as the target in the adjacent square, but this view nevertheless seems as a good approximation. This perspective is also discussed in \cite{Review-inter}.

The searcher starts at a random point of the torus,  and then moves according to some random walk strategy $X$. 
 In this strategy, the length of a step $\ell$ is chosen according to a specified distribution $p$, while its direction is chosen uniformly at random. 
In particular, for a given $\mu\in (1,3]$, a (truncated) {\em L\'evy walk} process $X^\mu$ on the torus $\T_n$ is a random walk whose step-lengths are distributed according to $p(\ell)\sim 1/\ell^\mu$, for $\ell\leq \sqrt{n}/2$. We discuss the influence of the choice of the cut-off later in the paper. For all processes, speed is assumed to be constant, hence the time duration of a step is proportional to its length. See more details in Methods.

  A {\em target} $S$ is a connected subset of the torus.  A  searcher can detect a target $S$ only when it is located within distance 1  --- the sensing range --- from the target. We consider several levels of detection that correspond to different abilities to detect targets while moving. The weakest is the {\em intermittent model}~\cite{Review-inter,reynolds2006intermittent}, which is especially relevant to the study of {\em saltatory}, or stop-and-go, foragers \cite{saltatory1,saltatory2,saltatory3}. In the intermittent setting, two modes of search alternate, and detection can only occur in one mode. In our intermittent model, one of these modes is static, corresponding to a short pause between ballistics steps where detected is enabled. Formally, the searcher detects a target, if and only if, at the end of a ballistic step, it is located at distance at most~$1$ from the target (see Fig~\ref{fig:illustration-model}). On the other extreme, we also consider the {\em continuous detection} model, in which the agent can detect a target also while moving, with a radius of detection~1. (Note that in the current paper, we focus on the time needed to find a single target, hence there is no need to specify whether the step is halted or not upon detection of a target, as in \cite{Viswanathan2}.)

 The  {\em detection  time} of a process $X$ with respect to $S$, denoted $t_{detect}^X(n,S)$, is the expected time until $X$ detects $S$ for the first time.  Expectation is taken with respect to the randomness of $X$ and the random initial location. We assume that the pause between ballistic steps takes a constant time.

 As we show, it turns out that the important parameter governing the detection time is not the area of $S$, but rather its diameter, namely, the maximal distance between any two points of $S$. Since the detection radius is 1, finding targets of smaller diameter  takes roughly the same time, hence, in what follows we assume that $D\geq 1$.

 To evaluate the search efficiency of  $X$ with respect to a  target $S$, we compare $t_{detect}^X(n,S)$ to $\opt(n,S)$, namely, the best achievable detection time of $S$.
 Importantly, when computing this optimal value, we impose no restriction on the search strategy, assuming the permissive continuous detection setting,  allowing the strategy to use infinite memory, and, furthermore, be tuned to the shape and the diameter of the target. 
  The following tight bound holds for every connected target $S$ whose diameter is $D\in[1,\sqrt{n}/2]$: 
\begin{linenomath*}
\begin{equation}
    \label{obs:lower}
  \opt(n,S)=\Theta\left({n}/{D}\right). 
\end{equation}\end{linenomath*}
The proof of Eq.~\eqref{obs:lower} appears in the Supplementary Materials, see Corollary 8. 
A sketch of the lower bound is given in Fig.~\ref{fig:illustration-lb}. For details regarding the asymptotic notation ``$\Theta$'', ``$O$'' and ``$\Omega$'', see Methods.

 We define the {\em overrun} of $X$ with respect to $S$, as an indicator of how well $X$ performs in comparison to the optimal algorithm:
\begin{linenomath*}\[ \Over^X(n,S)=\frac{t_{detect}^X(n,S)}{\opt(n,S)}
=\Theta\left ({t_{detect}^X(S)}\cdot \frac{D}{n}\right ).
\]\end{linenomath*}
 The {\em overrun} of $X$ with respect to a given diameter $D\geq 1$ is then defined  as the worst overrun, taken over all connected targets of diameter $D$, that is,  \begin{linenomath*}\begin{equation}    
 \Over^X(n,D)=\sup \{\Over^X(n,S) \mid S \mbox{~is of diameter~}D\}.
 \end{equation}\end{linenomath*} 
\begin{figure}[t]
  \centering
  \begin{subfigure}{0.49\textwidth}
  \centering
  \includegraphics[scale=0.2
  ,valign=t]{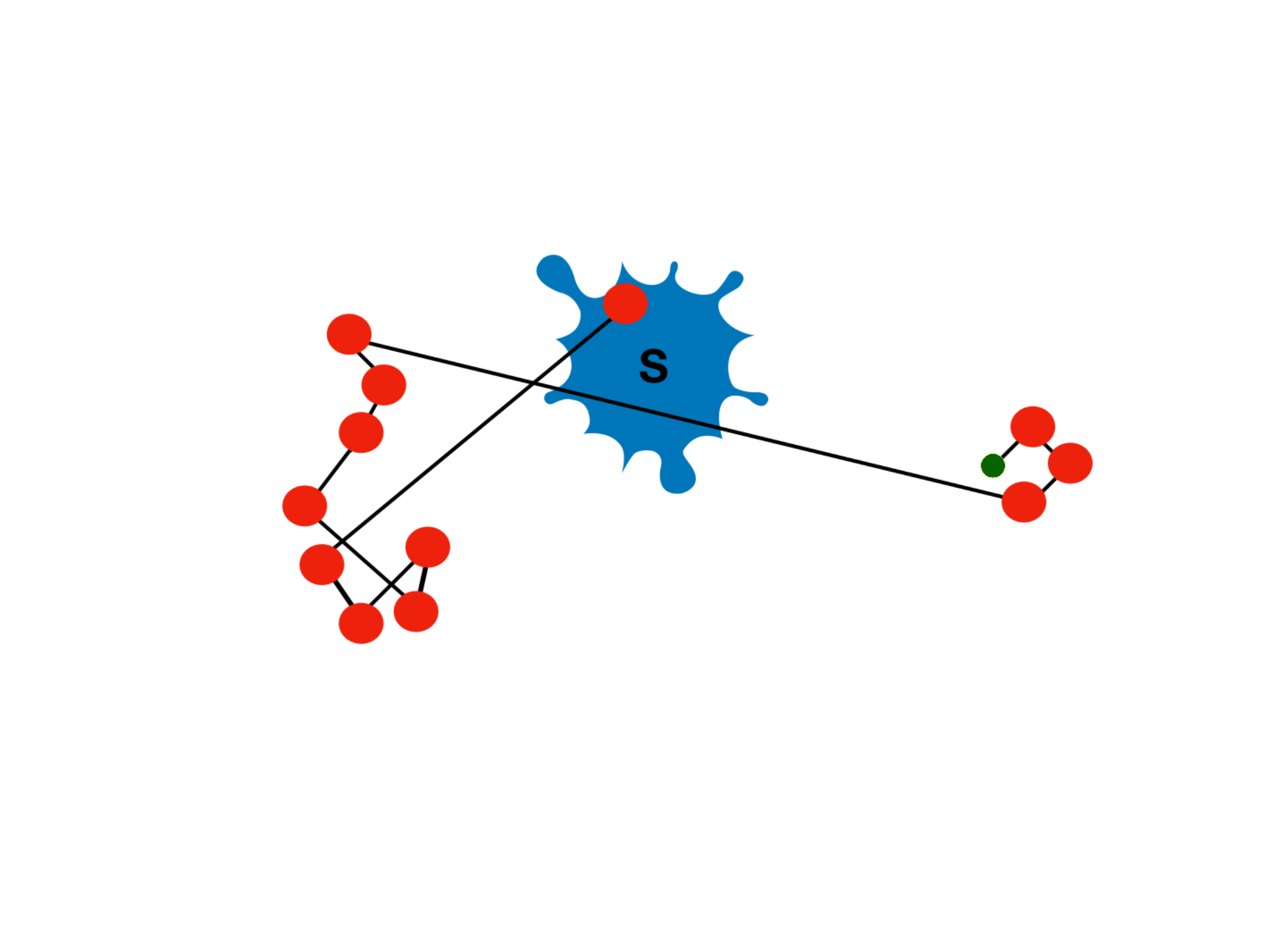}
\caption{}
\label{fig:illustration-model}
\end{subfigure}
 \begin{subfigure}{0.49\textwidth}
  \centering
 \includegraphics[scale=0.2
 ,valign=t]{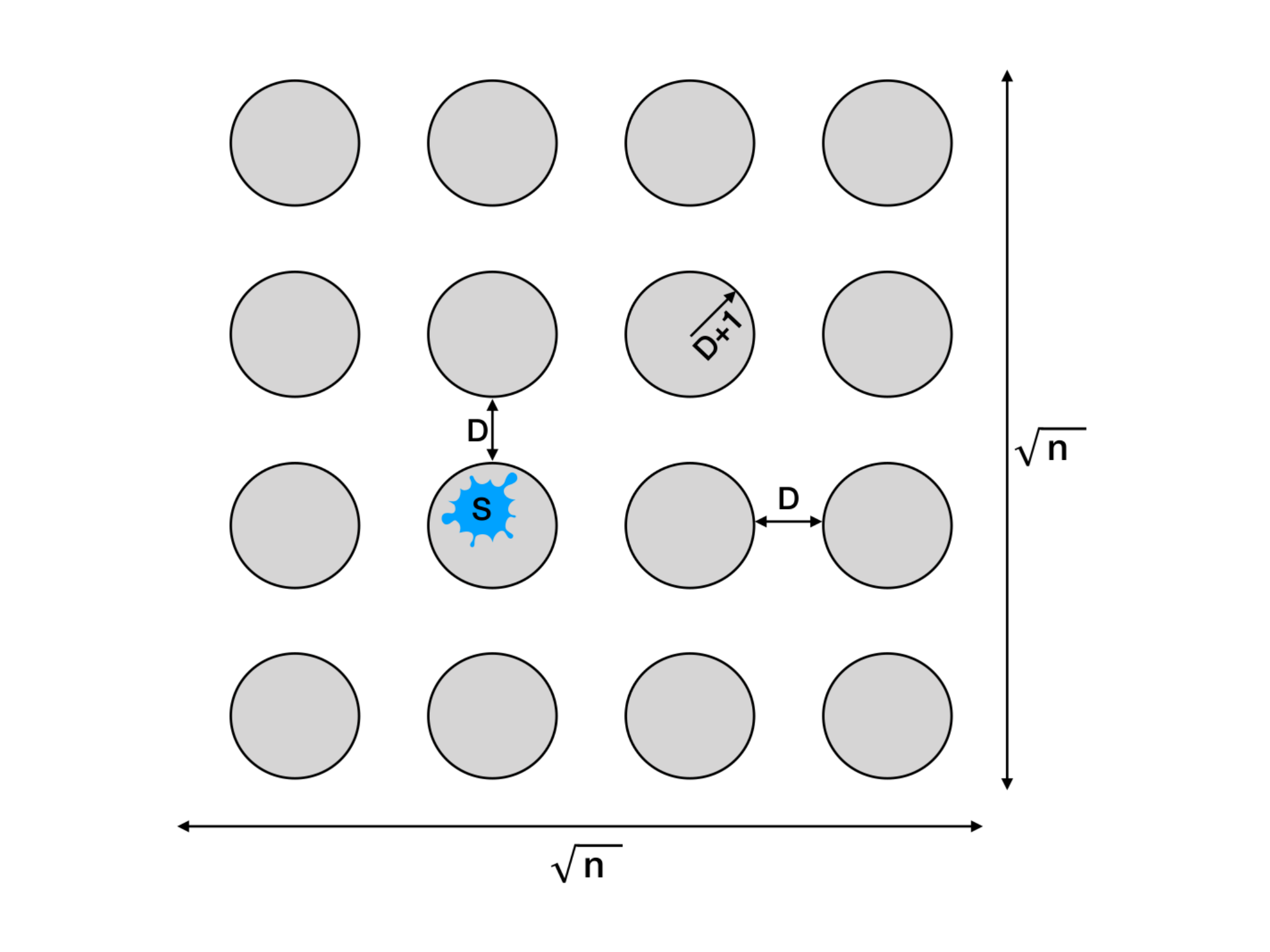}
 \caption{}
 \label{fig:illustration-lb}
 \end{subfigure}
\caption{(a) Intermittent L\'evy walk. The target $S$ is marked in dark blue. The L\'evy searcher starts at the smaller, green, point, and moves in discrete steps. A red circle signifies the area inspected at the end of a step - the disc of radius 1 around its location.  Here, the target $S$ is detected on the 12th step of the process. (b) Illustration of the lower bound proof of Eq.~\ref{obs:lower}. Consider a target $S$ (colored blue) of diameter $D$ (of any given shape). Consider roughly $n/(3D+2)^2$ discs (colored gray), so that each has radius $D+1$ and is located at distance $D$ from its neighboring discs. Furthermore, align this symmetric structure so that the $S$ touches the center of one of the discs.  Since the initial location of the searcher is uniform in the torus, with probability $\frac{1}{2}$, at least half of the discs need to be visited before detecting $S$. The time required to visit a new disc is at least the smallest distance between two discs, i.e., $D$. The detection time is therefore at least roughly $D\cdot n/(3D+2)^2\approx n/D$. }\label{fig:illustration}
\end{figure}
 In Supplementary Materials, Section B.1, 
 we demonstrate the definition of overrun, by providing a simple computation of the overrun of the intermittent process in which all step-lengths are fixed to some predetermined value. As seen there, such a strategy can be tuned to efficiently find targets of a particular size, however,  such an optimization causes inefficiency with respect to finding targets of other sizes. Hence, when targets appear in unpredictable sizes, it is unclear which intermittent strategy is best to employ.
\section*{Results}

\paragraph{The overrun of the Cauchy walk is poly-logarithmic for every target scale.}

We mathematically analyzed  the search efficiency of the intermittent Cauchy process $X^{\cauchy}$. We proved (Supplementary Materials, Section C, 
Theorem 18) 
that on the two-dimensional torus $\T_n$, the detection time of $X^{\cauchy}$ with respect to any target $S$ of diameter $D\geq 1$ is:
\begin{linenomath*}\begin{equation}\label{eq:theorem-upper}
    t_{detect}^{X^{\cauchy}}(n,S)=O\left( \frac{n\log^3 n}{D} \right).
\end{equation}\end{linenomath*}
The following result, which is an immediate corollary of Eq.~\eqref{eq:theorem-upper}, states that the overrun of the intermittent Cauchy walk with respect to {\em any} target diameter is poly-logarithmic in the size of the torus:
\begin{linenomath*}\begin{equation}\label{eq:cor-robust}
\mbox{For every $1\leq D\leq \frac{\sqrt{n}}{2}$,~~}\Over^{X^{\cauchy}}(n,D)=O(\log^3 n).
\end{equation}\end{linenomath*}
Eq.~\eqref{eq:cor-robust} is proved mathematically, and by its asymptotic nature, it holds for sufficiently large values of $n$. Using simulations (see Methods), we demonstrated that the overrun of the intermittent Cauchy walk is very small also for a relatively small domain (Fig.~\ref{fig:Cauchy_BallVsLine_n30x30}) and for a medium scale domain (Fig.~\ref{fig:Cauchy_BallVsLine_n300x300}). The overrun we see appears to be much smaller even from the poly-logarithmic upper bound of $O(\log^3 n)$. Indeed, detection time in $\T_{300^2}$ (Fig.~\ref{fig:Cauchy_BallVsLine_n300x300}) is very close to $2n/(D+1)$ for disc targets, and $4n/(D+1)$ for line targets.

As implied by Eq.~\eqref{obs:lower}, all connected targets of a given diameter $D$ share a common unconditional lower bound of $\Omega(n/D)$ for their detection time, regardless of their specific shape. Conversely, Eq.~\eqref{eq:theorem-upper} implies that such targets are found by roughly this time by the intermittent Cauchy process. These results suggest that, at least asymptotically, the right parameter to consider is indeed the diameter of the target and not, e.g., its area. We find this insight rather surprising, as, in contrast to a searcher in the continuous detection model, crossing the target's boundary by an intermittent searcher does not suffice for detection. Hence, for example, a disc-shaped target appears to be, at least at a first glance, significantly more susceptible for detection than its one-dimensional perimeter. Consistent with our claim, in Figs.~\ref{fig:Cauchy_BallVsLine_n30x30} and~\ref{fig:Cauchy_BallVsLine_n300x300} we see that the detection time of the intermittent Cauchy walk with respect to lines of diameter $D$ (orange curve) is only about twice larger than the detection time of a disc (blue curve) with the same diameter. This remains true even when the diameter is relatively large, e.g., $D=16$ in Fig.~\ref{fig:Cauchy_BallVsLine_n300x300},  despite the fact that  the area of the corresponding disc is more than $25$ times larger than the area of the domain from which a line of length $16$ can be detected, i.e., a strip of width $1$ and length $16$. A consequence of this insight suggests that a large prey aiming to hide from an efficient searcher would benefit by organizing itself in a bulging shape that minimizes its diameter.

\paragraph{Lower bounds.}
Eq.~\eqref{eq:cor-robust} establishes the small overrun of the Cauchy process across all target diameters. We next turn to study the
overrun of  L\'evy walk other than Cauchy (i.e., the cases $\mu\neq 2$). We proved (Supplementary Materials, Section B.3)  
that for $1<\mu<2$, the overrun of the corresponding intermittent L\'evy walk is large with respect to small diameter targets, and that for $2<\mu\leq 3$, the overrun is large with respect to large diameter targets.
The latter result holds also in the continuous detection model. 

In more details, we first considered the intermittent L\'evy walks with $1<\mu<2$, writing $\mu=2-\epsilon$, with $0<\epsilon<1$.
For these cases, it turns out that the expected step length is already polynomial in $n$, which means that the process is slow at finding small targets. 
Specifically, we proved (Supplementary Materials, Theorem 11) 
that the detection time of $X^\mu$ with respect to $S$ is:
\begin{linenomath*}\[t_{detect}^{X^\mu}(S)=\Omega(n^{1+\epsilon/2}/D^2). \]\end{linenomath*}
Dividing this lower bound by the unconditional optimal detection time of targets of diameter $D$, which is $\Theta(n/D)$, we obtain the following lower bound on the overrun of $X^\mu$:
\begin{linenomath*}\begin{equation}\label{eq:lower-levy-small}
    \Over^{X^\mu}(n,D)= \Omega(n^{\epsilon/2}/D).
    \end{equation}\end{linenomath*}
In particular, for targets with constant diameter, the overrun is  polynomial in $n$.

The lower bound established in Eq.~\eqref{eq:lower-levy-small} indicates that within the range $\mu\in(1,2)$, intermittent L\'evy walks with smaller values of $\mu$ (i.e., higher $\epsilon$) would lead to larger overrun, especially with respect to small diameter targets. Simulations reveal that this tendency is already apparent in small terrains  (Fig.~\ref{fig:LevyComp_n30x30_p0}, with  $n=30^2$). The tendency clearly sharpens for larger values of $n$, where the intermittent Cauchy walk can be seen to outperform  intermittent L\'evy walks with $\mu\in(1,2)$, for a large range of small target sizes (Fig.~\ref{fig:LevyComp_n300x300_p0}, with $n=300^2$).  

\begin{figure}
    \centering
    \begin{subfigure}{0.49\textwidth}
    \centering
    \includegraphics[scale=0.55]{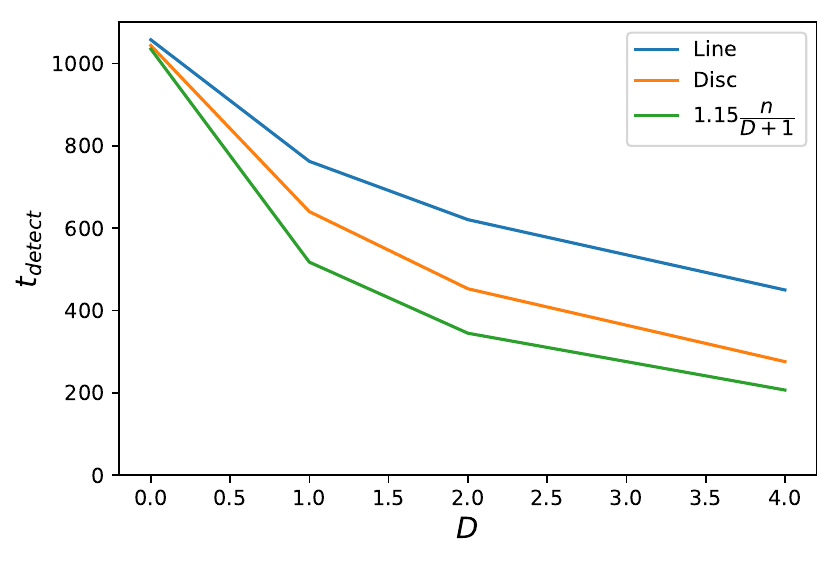}
    \caption{$n=30^2$}\label{fig:Cauchy_BallVsLine_n30x30}
    \end{subfigure}\begin{subfigure}{0.49\textwidth}
    \centering
    \includegraphics[scale=0.55]{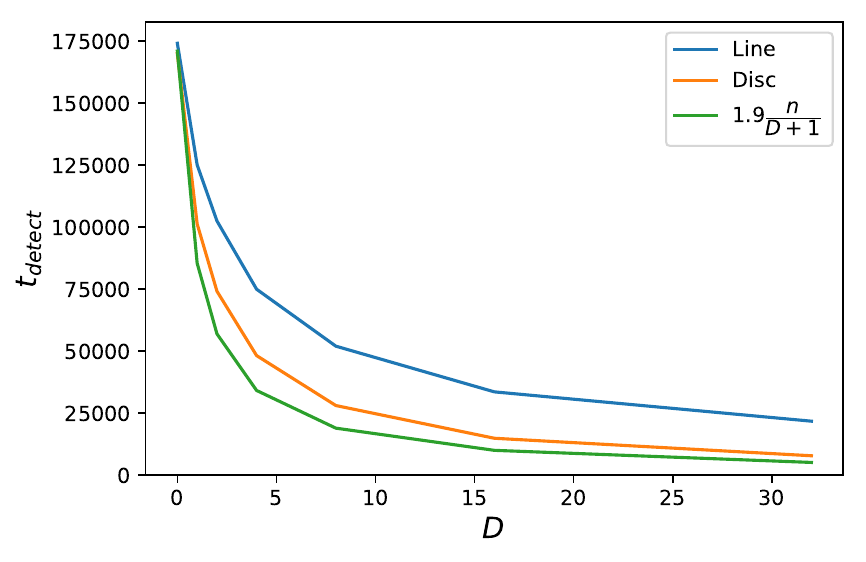}
    \caption{$n=300^2$}\label{fig:Cauchy_BallVsLine_n300x300}
    \end{subfigure}
    \caption{Detection time of the truncated Cauchy Walk on $\T_n$, searching for a disc (orange color) or line (blue color) target of diameter $D$. Green line is used for comparison.} 
    \label{fig:Cauchy_BallVsLine}
\end{figure}
Next, we consider the  L\'evy walks with $2<\mu\leq 3$, writing $\mu=2+\epsilon$ where $0<\epsilon\leq 1$. For this regime of $\mu$ we remove the intermittent assumption, allowing the strategy to perfectly detect at all times, i.e, we consider the continuous detection model. 
Intuitively, the lower bounds for these cases stem from the fact that such processes take long time to reach faraway locations. Hence, in comparison to the optimal strategy, these strategies are slow at finding large faraway targets. Specifically, we proved (Supplementary Materials, Theorem 12) 
that
\begin{linenomath*}\[t_{detect}^{X^\mu}(S)=\begin{cases}
 \Omega({n}{D^{\eps-1}}) \text{ if } \mu=2+\epsilon, \mbox{~where~} 0<\epsilon<1, \\
 \Omega(\frac{n}{\log D})  \text{ if } \mu=3. 
\end{cases}\]\end{linenomath*}
Again, dividing these lower bounds by  $n/D$, gives the following lower bounds:
\begin{linenomath*} \begin{equation}\label{eq:lower-comp-large-mu}
\Over^{X^\mu}(n,D)=\begin{cases}
 \Omega(D^{\eps}) \text{ if } \mu=2+\epsilon, \mbox{~where~} 0<\epsilon<1, \\
 \Omega(\frac{D}{\log D})  \text{ if } \mu=3. 
\end{cases}
  \end{equation}\end{linenomath*}
Comparing with intermittent L\'evy walks with $\mu\in(2,3]$, simulations demonstrate  that the intermittent Cauchy walk outperforms such walks with respect to almost all the range of target sizes, except for the very small ones (Figs.~\ref{fig:LevyComp_n30x30_p0} and~\ref{fig:LevyComp_n300x300_p0}). Moreover, the gap between the performances becomes larger when the target's diameter $D$ grows. This is consistent with the asymptotic bound in  Eq.~\eqref{eq:lower-comp-large-mu}.

\begin{figure}
    \centering
    \begin{subfigure}{0.49\textwidth}
    \centering
        \includegraphics[scale=0.55]{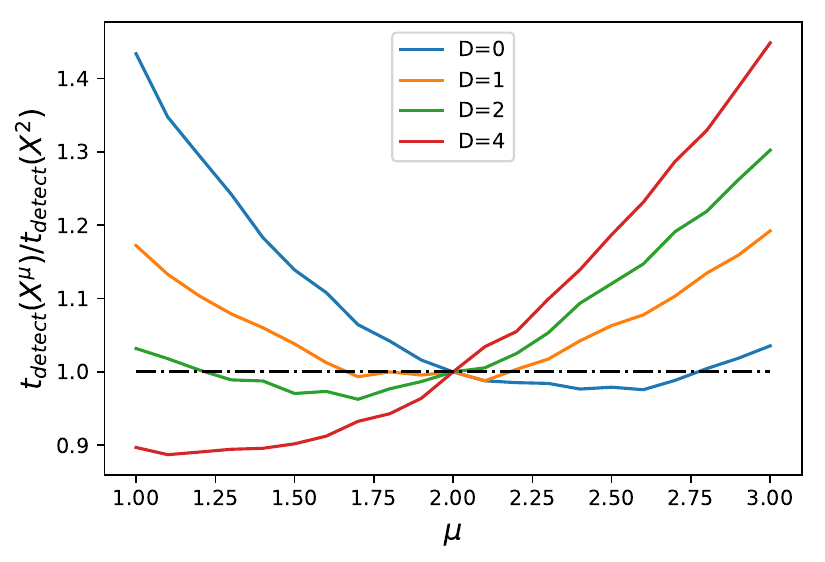}
        \caption{$n=30^2$, intermittent ($p=0$)}
        \label{fig:LevyComp_n30x30_p0}
    \end{subfigure}
    \begin{subfigure}{0.49\textwidth}
        \centering
        \includegraphics[scale=0.55]{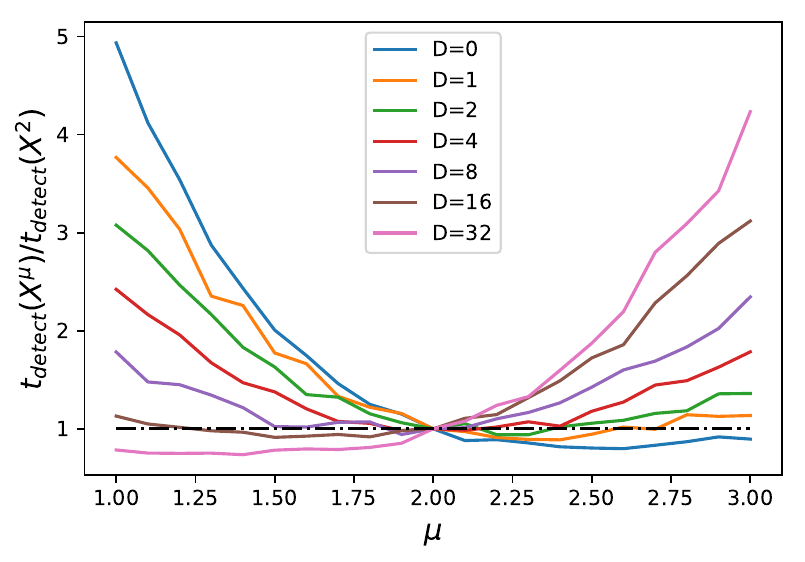}
        \caption{$n=300^2$, intermittent ($p=0$)}
        \label{fig:LevyComp_n300x300_p0}
    \end{subfigure}
    \begin{subfigure}{0.49\textwidth}    \centering
        \includegraphics[scale=0.55]{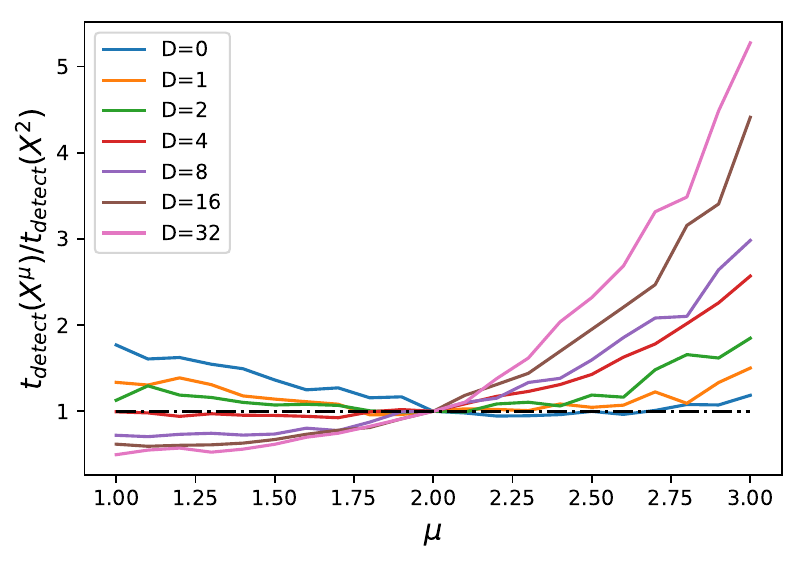}
        \caption{$n=300^2, p=0.1$}
        \label{fig:LevyComp_n300x300_p0dot1}
    \end{subfigure}
    \begin{subfigure}{0.49\textwidth}    \centering
        \includegraphics[scale=0.55]{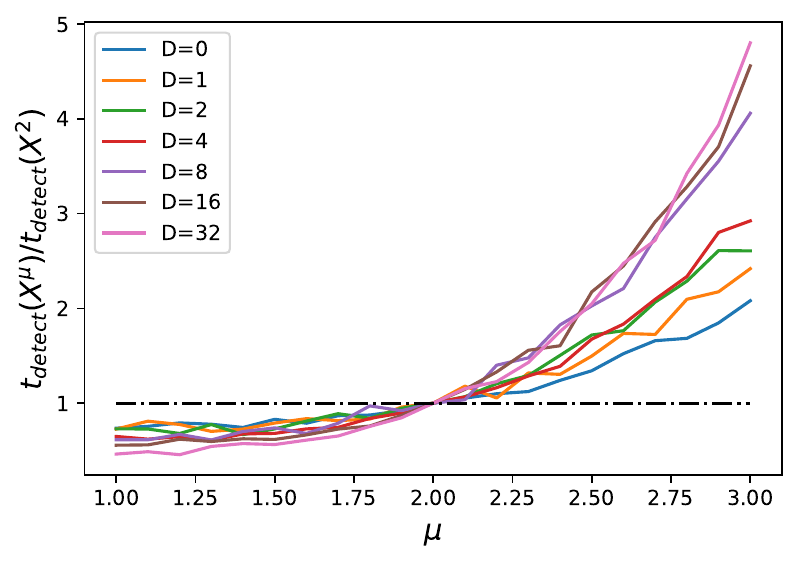}
        \caption{$n=300^2, p=1$}
        \label{fig:LevyComp_n300x300_p1}
    \end{subfigure}
    \caption{Comparing the detection times of L\'evy walks $X^{\mu}$ on $\T_n$, for different $\mu\in [1,3]$, with the detection time of the Cauchy walk ($\mu=2$). Search times are evaluated 
    with respect to disc targets of diameter $D$. For each diameter $D$, the data is normalized so that the detection time of the Cauchy Walk $X^2$ is represented by $1$. (a) and (b) consider the intermittent setting, on a relatively small torus of size $n=30^2$ (a), and a larger one of size $n=300^2$ (b). (d) considers the continuous, perfect, detection model, where the target is also  detected (with probability $p=1$) while moving ballistically, if the searcher is at distance at most $1$ from the target. 
    (c) considers the continuous, imperfect, detection model, where the target is detected, while moving, with probability $p=0.1$ for each unit of time that the searcher spends at distance at most $1$ of it, and, if the searcher is in-between steps (and located at distance at most $1$ from the target), then the target is detected with probability $1$.}
    \label{fig:LevyComp}
\end{figure}
\paragraph{On the impact of weak detection: intermittent vs. continuous.} 
 The intermittent detection model~\cite{Review-inter,reynolds2006intermittent,saltatory1,saltatory2,saltatory3} is motivated by the premise that scanning for targets is hard to effectively maintain continuously, and especially while moving fast \cite{bell2012searching,kramer2001behavioral,o1990search}. Many biological processes are considered to be intermittent, or at least partially so \cite{Review-inter}, however, the extent at which the detection is worsened by movement is often unclear.

The $O(\frac{n\log^3 n}{D})$ upper bound on the detection time of the Cauchy walk (Eq.~\eqref{eq:theorem-upper}) was established with respect to the intermittent setting. Clearly, it also holds when detection is strengthened. Since the bound on the optimal detection time, i.e., $\opt(n,D)=\Theta(n/D)$, holds also under the continuous detection model, it follows that the $O(\log^3 n)$ upper bound on the overrun of the Cauchy walk  (Eq.\eqref{eq:cor-robust}) is valid for all models of detection in-between intermittent and continuous detection. Furthermore, the established lower bounds  for $2< \mu\leq 3$ (Eq.~\ref{eq:lower-comp-large-mu}) hold also when detection is continuous. For $1<\mu<2$, however, the overrun lower bounds in Eq.~\eqref{eq:lower-levy-small} do not hold in the continuous detection model. Indeed, if detection occurs while moving, then previous simulations seem to indicate that a straight line movement, i.e., taking $\mu\approx 1$, is somewhat preferable~\cite{humphries2014optimal}.

To study the influence of the detection abilities while moving on the detection time, we also simulated the detection times of  L\'evy walks in continuous settings, in which detection while moving is weak, or imperfect, ($p=0.1$, Fig.~\ref{fig:LevyComp_n300x300_p0dot1}), and perfect ($p=1$, Fig.~\ref{fig:LevyComp_n300x300_p1}). 
Consistent with the theoretical results, the simulations reveal that 
the Cauchy walk outperforms L\'evy walks with $2<\mu\leq 3$ with respect to almost all the range of target sizes, and especially with respect to the larger targets, regardless of the detection ability while moving. 

On the other hand, for $1<\mu<2$, the overrun with respect to small targets is significantly improved when detection while moving is strengthened. Indeed, in the continuous, perfect, detection model ($p=1$, Fig.~\ref{fig:LevyComp_n300x300_p1}), we find that regardless of target size, detection is faster when $\mu$ tends to 1, as expected. In the continuous, imperfect, detection model ($p=0.1$, Fig.~\ref{fig:LevyComp_n300x300_p0dot1}), the situation is intermediate between the perfect and the intermittent settings.

\paragraph{On the influence of the cut-off.}
We first note that having a cut-off is reasonable for biological applications, which live in finite domains. Moreover, from a theoretical perspective, in contrast to the continuous detection model \cite{Viswanathan2}, the intermittent setting forces L\'evy walks with $\mu\leq 2$ to come with a cut-off, as otherwise the expected length of a step would be infinite, implying infinite expected time to find any target. As a result of the truncation, the variances of the processes we consider are also finite. However, as we proved in the Supplementary Materials (Lemma 20),
the super-diffusive property of the Cauchy walk, which was used to derive the upper bound on its detection time, still holds at least up to time $\Theta(\sqrt{n})$.

Note that by the nature of our asymptotic results, the upper bound on the overrun of the Cauchy walk (Eq.~\eqref{eq:cor-robust}) is expected to hold when taking the cut-off $\lmax=\Theta(\sqrt{n})$. Therefore, for sufficiently large values of $n$, an efficient Cauchy strategy needs only to be loosely tuned to the size of the domain.

To quantify the influence of the cut-off on moderate size domains, we simulated the Cauchy walk with different cut-offs on the torus $\T_n$ where $n=300^2$. For different diameters, Fig.~\ref{fig:Cauchy_cutoff_influence} depicts a comparison between the performances of the Cauchy walk with cut-off $\lmax=\sqrt{n}/2=150$ and those with cut-offs $\lmax\in [112,1200]$. 
 Observe that over-estimating the area $n$ of the domain by a factor $64$ (or, equivalently, its diameter by a factor $8$) does not lead to a drastic change in performance. Indeed, these Cauchy walks perform at most $1.4$  times worse than  the Cauchy walk with cut-off $\lmax=\sqrt{n}/2$. This is significantly  less than the relative values observed for other L\'evy walks in Fig.~\ref{fig:LevyComp_n300x300_p0}. We conclude therefore that the Cauchy walk performances are not very sensitive to the value of the cut-off $\lmax$. Indeed, intuitively, for $\lmax=\Omega(\sqrt{n})$, the dependency of the time performances of the Cauchy walk on $\lmax$ is logarithmic, as the average length of a step is $\Theta(\log \lmax)$.

\begin{figure}
    \centering
    \includegraphics[scale=0.55]{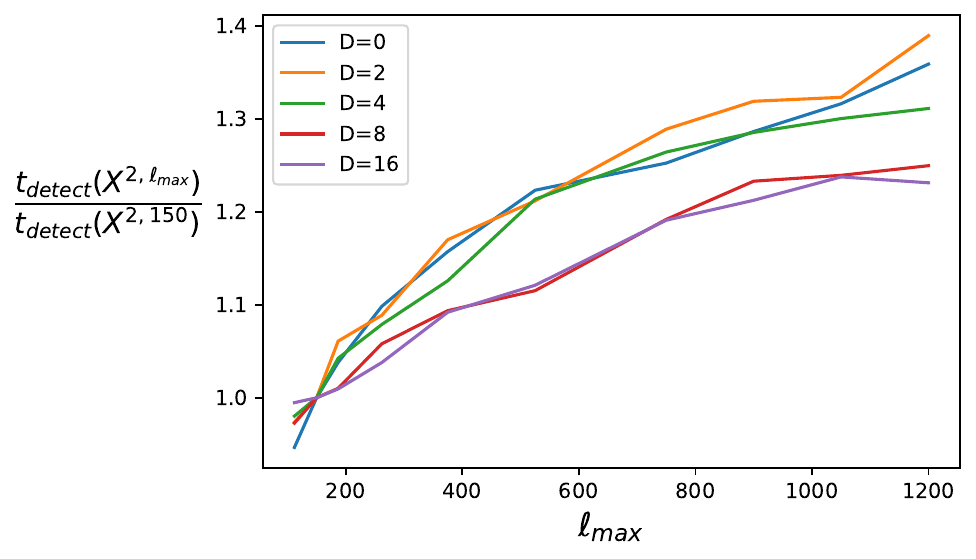}
    \caption{Effect of the cut-off $\lmax$ on the detection time of the Cauchy Walk on $\T_n$ with $n=300^2$. The target is a disc of diameter $D$. The plot is normalized with a value of $1$ for the cut-off $\lmax=\sqrt{n}/2=150$.}
    \label{fig:Cauchy_cutoff_influence}
\end{figure}

\section*{Discussion}
This paper evaluates search strategies according to their efficiency in finding targets of varying sizes \cite{peres2015mixing}. This measure is motivated by the fact that in  multiple foraging contexts, including ones for which L\'evy patterns have been reported, targets appear in varying sizes. Importantly, quickly finding targets of all sizes means that areas of all scales are visited quickly and regularly. This has significance also in other tasks than foraging, including, e.g., during eye scanpaths \cite{eyes}, viral spreading \cite{Janssen}, and movement of metastatic cancer cells \cite{metastatic}. For all these examples, intermittent patterns are of interest and L\'evy walk movement has been suggested.

We further stress that target size in the sense we consider here concerns not the physical size of the target, but rather its {\em effective size}, corresponding to the area from which it can be detected. The effective size of a target is impacted not only by its physical size, but also by the detection abilities of the searcher with respect to the environmental conditions at the vicinity of the target. For example, a rabbit in flat open space can be located from a farther distance than if it were located in a bushy area. Similarly, an eye searching for a red spot in the visual field could detect it from a larger distance if the background were, e.g.,  blue instead of pink. Thus, even when the physical size of the target is fixed, its effective size can vary.
 In our mathematical analysis, we normalized detection radius to 1, and allow for varying target sizes.  We note, however, that this modelling can also capture varying detection radii. Indeed, if the actual detection radius is $R>1$, and the physical diameter of the target is $D$, then the situation is equivalent to searching for a target of diameter roughly $D+2R$ using detection radius of 1. Hence, the established robustness of the Cauchy walk with respect to all target scales also implies robustness to both varying target scales and varying detection radii.

As proven here, intermittent Cauchy walks are almost optimal when the goal is to quickly find sparse targets of unpredictable sizes (or when the detection radius varies). Compared to L\'evy walks with $2<\mu\leq 3$, the performances of the Cauchy walk are particularly advantageous with respect to larger targets. This superiority remains true regardless of whether the detection is intermittent or not. On the other hand, compared to L\'evy walks with $1<\mu< 2$, the striking superiority of the Cauchy walk holds only when the search is intermittent.
 These results shed a new light on the L\'evy foraging hypothesis \cite{viswanathan2008levy}, and can thus initiate new directions for experimental work on animals suspected to perform L\'evy walks. One suggestion is to experimentally study the correlation between (1)  the distribution of target sizes \cite{MarinePredatorEnvironment,weimerskirch2007seabirds}, (2) the exponent $\mu$ of the corresponding L\'evy walk, and (3) the animal's detection abilities. In contexts where the L\'evy searcher aims to quickly find targets of varying sizes, we predict that the exponent $\mu$ will not be much higher than~$2$. This, for example, is consistent with the albatrosses  foraging on heterogeneous patches of squid and fish \cite{weimerskirch2007seabirds}, whose L\'evy movement patterns were estimated to have an exponent of $\mu\approx 1.25$ \cite{humphries2012foraging}.  Moreover, if, in addition, the L\'evy searcher relies on deficient detection while moving, then we predict that $\mu$ will tend to be closer to 2, giving rise to a Cauchy walk. This is consistent with fruit flies whose exploration trajectories   were reported to be both intermittent  and Cauchy \cite{reynolds2007free}. 
 Accordingly, it is worth inspecting whether other biological searchers that have been identified as executing Cauchy movement patterns, including multiple species of marine predators \cite{MarinePredator,MarinePredatorEnvironment,bartumeus2003helical},  T-cells \cite{Tcell}, and honey bees \cite{reynolds2007displaced}, have poor detection abilities while moving.
 
 To conclude, until now there was no rigorous explanation for the superiority of L\'evy walks in dimensions higher than one. This paper is the first to provide such an explanation. First, we prove  that in finite two-dimensional domains, (truncated) Cauchy walks find sparse targets of any size in almost optimal time. Moreover, under intermittent detection, any other L\'evy walk fails to efficiently find both small and large targets. This highlights the impact of weak detection on the incentive to perform Cauchy walks.

\section*{Methods}\label{sec:methods}

\subsection*{Model.} Detailed analytical proofs of the results mentioned in the main text are presented in the Supplementary Materials. We next provide further details on the model, complementing the ones mentioned in the main text.

We consider a mobile agent that searches a target over 
 the finite torus $\T_n$ identified as the set $[-\sqrt{n}/2,\sqrt{n}/2]^2$ in $\R^2$. Note that the area of the torus is  $n$.  For $x=(x_1,x_2)\in \Omega$, we consider the standard norm $\norm{x}=\sqrt{x_1^2+x_2^2}$.

 We consider a general family of random walk processes, composed of discrete randomly oriented ballistic steps. In these strategies, the length of a step $\ell$ is chosen according to a specified distribution $p$,  while its direction is chosen uniformly at random. More precisely, a {\em random walk process} on  $\T_n$ 
 is a process $X$ such that the initial position $X(0)$ is given by a uniform distribution 
 and for every integer $m\geq 0$,
\begin{linenomath*}\[\quad X(m+1)=X(m)+V(m+1),\]\end{linenomath*}
 where $(V(m))_{m\geq 1}$ are the independent and identically distributed (i.i.d) steps. The sum $X(m)+V(m+1)$ is taken modulo the torus $\T_n$. 
 The lengths of  $\ell=\norm{V(m)} $ of the steps are chosen according to some distribution $p(\ell)$, and the angle of each step is chosen uniformly at random.

 A {\em L\'evy walk}  $X^\mu$ on $\T_n$,  for a given $\mu\in (1,3]$ and maximal step $\ell_{max}=\sqrt{n}/2$, is the random walk process whose step-lengths are distributed according to
\begin{linenomath*}\begin{equation}\label{eq:cauchy-precise-distribution}p(\ell)=\begin{cases}a \text{ if } \ell\leq 1\\ a \ell^{-\mu}\text{ if } \ell \in (1,\lmax) \\
0 \text{ if } \ell \geq \lmax \end{cases},
\end{equation}\end{linenomath*}
where $a=(1+\int_{1}^{\lmax} \ell^{-\mu}d\ell)^{-1}$ is the normalization factor. 
Note that as $\mu$ grows from 1 to 3, the behaviour changes from being almost ballistic to being diffusive-like \cite{Viswanathan2}. 
When $\mu=2$, we refer to the process as a {\em Cauchy walk}. The Cauchy walk on the torus is denoted $X^{\cauchy}$.
For all processes, speed is assumed to be constant. Specifically, doing a step of length $\ell$ necessitates $\Theta(\ell)$ time units.
The scanning time is some constant $b$. Hence, the time used to take a step of length $\ell$ (including the scanning time before the step starts) is  $\Theta(\ell)+b$.

For an integer $m$, the random time $T(m)$ taken by the walk up to step $m$ is defined as \begin{linenomath*}\[ T(m)=\sum_{s=1}^m (\norm{ V(s) } +b).\]
\end{linenomath*}
As we see in the Supplementary Materials (Section A.1),
the average  length $\tau$ of a ballistic step is at least some constant. This implies that the average time spent during $m$ steps (including the scanning time), is  proportional to $m\tau$.  

\subsection*{Asymptotic notation.}

We adopt the Bachmann-Landau classical mathematical asymptotic notation (see Chapter 3 in \cite{notation}). These notations describe the limiting behaviour of functions as their 
argument, which is in our case the size of the torus $n$,  tends towards infinity. 
Specifically, consider two non-negative function $f$ and $g$ defined on the integers. The ``$O$'' notation represents an upper bound in the following sense. We say that $f(n) \in O(g(n))$ if there exists $c > 0$ and an integer $n_0$ such that $f(n) \le c \cdot g(n)$ for all $n \ge n_0$. Conversely, the asymptotic lower bound notation ``$\Omega$'' is interpreted as follows. We say that $f(n) \in \Omega(g(n))$ if there exists a constant $c > 0$ and an integer $n_0$ such that $c \cdot g(n) \le f(n)$ for all $n \ge n_0$. 
Finally, the ``$\Theta$'' notation represents a tight asymptotic bound (up to constant factors). Specifically, $f(n) \in \Theta(g(n))$ if both $f(n) \in O(g(n))$ and $f(n) \in \Omega(g(n))$. 
\subsection*{Simulations.}
Using Python, we simulated an agent performing a L\'evy walk starting at a point uniformly at random in the torus $\T_n$, searching for a target of diameter $D$ located at the center of the torus. The L\'evy distribution was approximated by its discrete equivalent $p(\ell)=a_{\mu,\lmax}\ell^{-\mu}$ for $\ell \in \{1,\dots,\lceil \lmax \rceil \}$. Aside from Figure \ref{fig:Cauchy_cutoff_influence}, we took $\lmax=\sqrt{n}/2$. 1000 runs were performed for each couple $(\mu,D)$.

\bibliography{largetargets}

\subsection*{Acknowledgments.} The authors are thankful to Ofer Feinerman for helpful discussions, and to Robin Vacus for commenting on the presentation. This work has received funding from the European Research Council (ERC) under the European Union's Horizon 2020 research and innovation program (grant  agreement No 648032).

\subsection*{Data availability statement.} Complete proofs for the theorems can be found in the Supplementary Materials. The code for reproducing the simulations can be found at \url{ https://github.com/BrieucZambrano/levy-walks}.

\subsection*{Author contribution.} Both authors contributed equally on the analysis and conceptualization. B.G. conducted the simulations, and A.K. wrote the main text.

\subsection*{Competing interests.} Authors declare having no competing interests.

\subsection*{Correspondence.}  A.K is the corresponding author, email: amos.korman@irif.fr.

\clearpage

\appendix
\centerline{\huge Supplementary Materials}

\section{Preliminary theoretical results}
For general definitions regarding the model,  see Methods in the main text.
Let us, however, recall here few definitions that will be used  extensively.

The torus  $\T_n$ is identified with the set $[-\sqrt{n}/2,\sqrt{n}/2]^2$ in the infinite plain $\R^2$.
Consider $\mu\in (1,3]$ and maximal step $\ell_{max}>1$ (possibly $\ell_{max}=\infty$). A {\em L\'evy walk} $Z^\mu$ on $\R^2$ (or $X^\mu$ on $\T_n$),  with maximal step $\ell_{max}>1$,  is the random walk process whose step-lengths are distributed according to
\begin{linenomath*}\begin{equation}\label{eq:cauchy-precise-distribution-appendix}p(\ell)=\begin{cases}a \text{ if } \ell\leq 1\\ a \ell^{-\mu}\text{ if } \ell \in (1,\lmax) \\
0 \text{ if } \ell \geq \lmax \end{cases},
\end{equation}\end{linenomath*}
where $a=(1+\int_{1}^{\lmax} \ell^{-\mu}d\ell)^{-1}$ is the normalization factor. When considering a L\'evy process on the torus, we shall take $\ell_{max}=\sqrt{n}/2$. 
Recall also, that when $\mu=2$, we refer to the process as a {\em Cauchy walk}. The Cauchy walk on the torus is denoted $X^{\cauchy}$.

In addition, we shall extensively use the following definition.
\begin{definition}
Given a target $S$, the {\em extended set} $B(S)$ is the set of nodes at distance at most 1 from $S$. Note that since the radius of detection is 1, the searcher detects $S$ if and only if it is located in $B(S)$. 
\end{definition}
\subsection{Expectations and variances of step-lengths}\label{sec:average}
\begin{claim}\label{claim:exp+var}
Consider the L\'evy walk $Z^{\mu}$ (or $X^{\mu}$) with maximal step length $\ell_{max}$. The average length  of a step (and hence the average time to take a step) is 
\begin{linenomath*}\begin{equation}\label{obs:tau}
\tau=\begin{cases} \Theta(\lmax^{2-\mu}) \text{ if } \mu\in (1,2) \\
\Theta(\log \lmax) \text{ if } \mu=2 \\
\Theta(1) \text{ if } \mu\in (2,3]\end{cases},
\end{equation}\end{linenomath*}
and the variance $\sigma^2$ and second moment $M$ of a step-length are
\begin{linenomath*}\begin{equation}\label{obs:var}
\sigma^2=\Theta(M)=\begin{cases}\Theta(\lmax^{3-\mu}) \text { if } \mu\in(1,3) \\
\Theta(\log \lmax ) \text{ if } \mu=3\end{cases}.
\end{equation}\end{linenomath*}
\end{claim}
\begin{proof}
Given the definition of $p$, the expected step-length  is
\begin{linenomath*}\[ \tau=\int_0^1 a\ell d\ell + \int_1^{\lmax} a\ell^{1-\mu}d\ell .\] \end{linenomath*}
The first term is $\frac{a}{2}$, a constant, the second term is $\Theta(\lmax^{2-\mu})$ if $\mu\neq 2$, and $\Theta(\log \lmax)$ if $\mu=2$. 
The second moment $M$ is computed likewise:
\begin{linenomath*}\begin{align*} M &=\int_0^{\lmax} \ell^2p^{\mu}(\ell)d\ell=\int_0^1 a\ell^2d\ell+\int_1^{\lmax} a\ell^{2-\mu} d\ell.
\end{align*}\end{linenomath*}
We have $\int_0^1 a\ell^2d\ell=\frac{a}{3}$ 
 for the first term, and for the second term
\begin{linenomath*}\[\int_1^{\ell_{max}} \ell^{2-\mu} d\ell = \begin{cases} \Theta(\ell_{max}^{3-\mu}) \text{ if } \mu <3 \\
\Theta(\log(\ell_{max})) \text{ if } \mu =3
\end{cases}.\] \end{linenomath*}
Now remark that $\tau^2=o(M)$, so that $\sigma^2=\Theta(M)$. 
\end{proof}

\subsection{On the connection between time and number of steps}\label{sec:wald}

To ease the notation, we drop the dependency on $n$ in several notations when it is clear from the context. Recall that we assume that the scan phase in-between ballistic step takes $\tau_{scan}=O(1)$ time. We next observe, that we may assume without loss of generality that this phase takes zero time, rather that a constant. Indeed, Claim \ref{claim:timeVSmoves} connects the detection time with the expected number of moves times the expected length of a step. If we take into consideration that the duration of the scan phase is $\tau_{scan}=O(1)$, then we would need to multiply the expected number of moves by the average time to take a step (including the pause before it) which is $\tau+\tau_{scan}$ instead of by $\tau$. As shown in Claim \ref{claim:exp+var}, we have  $\tau=\Omega(1)$ and thus $\tau+\tau_{scan}=\Theta(\tau)$. This implies that the asymptotic detection time is not affected by assuming that $\tau_{scan}=0$.

Let us denote by $T(m)$ the random time taken by the walk up to step $m$, i.e. \begin{linenomath*}\[ T(m)=\sum_{s=1}^m \norm{ V(s) } , \] \end{linenomath*}
where $V(s)=(V_1(s),V_2(s))$ is the vector chosen at step $s$, 
and $\norm{ V(s) } = \lvert V_1(s)\rvert + \lvert V_2(s) \rvert $.
Let us denote by $m^X_{detect}(S)$ the random number of steps before $X$ detects $S$ for the first time (i.e., since the searcher has a perception radius $1$, $m^X_{detect}(S)$ is the first $m$ such that $X(m)\in B(S)$). By definition, the expected time before detecting $S$ is  $t_{detect}^X(S)=\E(T(m_{detect}^X(S))$. We next argue that this time equals the average number of steps needed to hit $S$, multiplied by the average time $\tau$ needed for one step.

\begin{claim}\label{claim:timeVSmoves}
For any intermittent random walk $X$ on $\T_n$, and any set $S\subseteq \T_n$, \begin{linenomath*}\[t_{detect}^X(S) = \E(m_{detect}^X(S)) \cdot\tau,\] \end{linenomath*}
where $\tau=\E(\norm{V(1)})$ is the expected step-length.
\end{claim}
Claim \ref{claim:timeVSmoves} reminds of Wald's identity with respect to the lengths $(\norm{  V(s) })_s$. However, Wald's identity cannot be applied directly because $m^X_{detect}(S)$ is not a stopping step\footnote{The usual terminology is {\em stopping time}, but we employ the term ''step'' here so as to emphasis that the variable counts steps.}  for the sequence $(\norm{  V(s) })_s$. 
Instead, we prove the claim by the Martingale Stopping Theorem (that can also be used to prove Wald's identity). 

\begin{proof}
To prove the claim, note that we can suppose that $\tau<\infty$ and $\E(m^X_{detect}(S))<\infty$. Indeed, if $\tau=\infty$, then even one step takes an infinite expected time. 
Moreover, since $p(0)<1$ by definition, there exist $\epsilon,\delta>0$ such that the probability that a length of a step is at least $\epsilon$ is at least $\delta$. 
If $\E(m^X_{detect}(S))=\infty$, then, after $m$ steps, where $m$ is large, there are roughly $\delta m$ steps of length at least $\epsilon$. Hence, if there is an infinite number of steps, then with probability $1$ there is an infinite number of steps, each of which taking time at least $\epsilon$. In both cases, we have $t_{detect}^{X}(S)=\infty$, and the equality is verified. In what follows we therefore assume that both $\tau<\infty$ and $\E(m^X_{detect}(S))<\infty$.

We start the proof by defining: \begin{linenomath*}\[ W(m):=\sum_{s\leq m} (\norm{ V(s) } -\tau ).\] \end{linenomath*} The claim is proven by showing first that $( W(m))_m$ is a martingale with respect to $(X(m))_m$. Then, as $m^X_{detect}(S)$ is a stopping step for $(X(m))_m$ (i.e., the event $\{m^X_{detect}(S)=m\}$ depends only on $X(s)$, for $s\leq m$), we can apply the Martingale Stopping Theorem which gives $\sum_{s\leq m^X_{detect}(S)} (\norm{ V_s} -\tau)= 0$.
In more details, recall 
that a sequence of random variables $(W(m))_m$ is a {\em martingale} with respect to the sequence $(X(m))_m$ if, for all $m\geq 0$, the following conditions hold:
\begin{itemize}
    \item $ W(m)$ is a function of $X(0), X(1), \dots, X(m)$,
\item $\E(\lvert  W(m)\rvert )<\infty$,
\item $\E( W(m+1) \mid X(0),\dots, X(m))  =  W(m)$. 
\end{itemize}
We first claim that $W(m)$ is a martingale with respect to $X(0),X(1),\ldots$. Indeed, since $V(s)= X(s)-X(s-1)$, the first condition holds. Since $\E(\lvert  W(m)\rvert )\leq \sum_{s\leq m} \E(\lvert V_s-\tau \rvert ) \leq 2\tau m <\infty$, the second condition holds. Finally, since $ W(m+1)= W(m)+\norm{ V(m+1)} -\tau $,  we have $\E( W(m+1) \mid X(0),\dots, X(m))= W(m)+\E(\norm{ V(m+1)})-\tau= W(m)$, and hence the third condition holds as well.

Next, recall the Martingale Stopping Theorem 
which implies that $\E( W(M))=\E( W(0))$, 
whenever the following three conditions hold:
\begin{itemize}
    \item $ W(0),  W(1), \dots$ is a martingale with respect to $X(0), X(1), \dots$,
    \item $M$ is a stopping step for $X(0), X(1), \dots$ such that $\E(M)<\infty$, and 
    \item there is a constant $c$ such that $E(\lvert  W(m+1) -  W(m) \rvert  \mid X(0),\dots, X(m) ) < c$.
    \end{itemize}
Let us prove that the conditions of the Martingale Stopping theorem hold. We have already seen that the first condition holds. Secondly, we have $\E(m^X_{detect}(S))<\infty$ by hypothesis. Finally, we need to prove that $\E(\lvert  W(m+1) -  W(m) \rvert \mid X(0),\dots, X(m) )< c$ for some $c$ independent of $m$. Since $ W(m+1) -  W(m)=\norm{ V(m+1)}-\tau $, we have $\E(\lvert  W(m+1) -  W(m) \rvert \mid X(0),\dots, X(m) ) = \E(\big \lvert \norm{ V(m+1)}-\tau \big \rvert ) \leq 2\tau$. Therefore, the conditions hold and the theorem gives:
\begin{linenomath*}\[\E( W(m_{detect}(S)))=\E( W(0))=0.\] \end{linenomath*}
Hence, 
\begin{linenomath*}\begin{align*} 0&=\E( W(m^X_{detect}(S))) =  \E\left( - m^X_{detect}(S)\tau+\sum_{s\leq m^X_{detect}(S)} \norm{ V(s)} \right) \\& =  - \E(m^X_{detect}(S)) \tau +\E\left(\sum_{s\leq m^X_{detect}(S)} \norm{ V_s} \right) \\ 
& = - \E(m^X_{detect}(S)) \tau + t^X_{detect}(S), \end{align*}\end{linenomath*}
which establishes Claim \ref{claim:timeVSmoves}.
\end{proof}

\subsection{Monotonicity}\label{app:mon}
A function $f$ on $\R^2$ is called {\em radial} if there is a function $\tilde{f}$ on $\R^{+}$ such that for any $x\in \R^2$, $f(x)=\tilde{f}(\norm{ x})$. In this case we say that $f$ is non-increasing if $\tilde{f}$ is. 
The goal of this section is to prove the following.

\begin{claim}\label{cor:monotonicity-variables}
Let $X$ and $Y$ be two independent random variables with values in $\R^2$, admitting probability density functions respectively $f$ and $g$. Let $h$ be the probability density functions of $X+Y$. If $f$ and $g$ are both 
radial and non-increasing functions then so is $h$.
\end{claim}
We shall soon prove the claim, but first, let us give a corollary, assuming the claim is true.
\begin{corollary}[Monotonicity]\label{cor:monotonicity}
Let $Z$ be a random walk process on $\R^2$, starting at $Z(0)=0$, with step-length distribution $p$. If $p$ is non-increasing, then for any $m\geq 1$ the distribution $p^{Z(m)}$ of $Z(m)$ is radial and non-increasing. In particular, for any $x,x'$ points in $\R^2$ with $\norm{x}\leq \norm{x'}$, we have $p^{Z(m)}(x')\leq p^{Z(m)}(x)$. Furthermore, for any $x\in \R^2$ and any $m\geq 1$, $p^{Z(m)}(x) \leq \frac{1}{\pi\norm{ x}^2}.$\end{corollary}

\begin{proof}The fact that $p^{Z(m)}$ is radial and non-increasing follows from Claim \ref{cor:monotonicity-variables} by induction. Indeed, the step-length vectors $V(1), V(2),\dots$ are independent and, by hypothesis, admit a radial, non-increasing p.d.f. Hence so does $Z(m)=V(1)+V(2)+\dots+V(m)$.
The upper bound on $p^{Z(m)}(x)$ follows easily. Indeed, for $x\in \R^2\setminus \{(0,0)\}$, consider the ball $B$ of radius $\norm{x}$ and centered at $0$. We have $\int_{B}p^Z_m(y)dy\leq 1$, and by the monotonicity, $\int_{B}p^Z_m(y)dy\geq p^Z_m(x) \lvert B\rvert = p^Z_m(x) \cdot \pi \norm{ x}^2$.
\end{proof}

\begin{proof}[Proof of Claim \ref{cor:monotonicity-variables}]
Let $\theta\in [0,2\pi)$. For $x\in \R^2$, denote by $rot_{\theta}(x)$ the point obtained by  rotating $x$ around the center $0$ with an angle of $\theta$. Then, by a change of variable, we have:
\begin{linenomath*}\begin{align*} h(rot_{\theta}(x))&=\int_{y\in \R^2} f(rot_{\theta}(x)-y)g(y)dy \\
&=\int_{y\in \R^2} f(rot_{\theta}(x)-rot_{\theta}(y))g(rot_{\theta}(y))dy\\
&= \int_{y\in \R^2} f(x-y)g(y)dy=h(x), \end{align*}\end{linenomath*}
where we used in the last equality the radiality of $f$ and $g$. 
This establishes the fact that $h$ is radial. Next, we prove, in a manner inspired by Adler et al.~\cite{hunter}, 
that $h(x)$ is non-increasing with $\norm{ x}$. Since $h$ is radial, we can restrict the study to points of the non-negative $y$-axis. Let us fix $x=(0,x_2)\in \R\times\R^{\geq 0}$, and $x'=(0,x_2')\in \R\times\R^{\geq 0}$ with $x_2'\geq x_2$. Our goal is to show that $h(x)\geq h(x')$. 

Let $\gamma=\frac{x'_2-x_2}{2}$. Note that $f(0,x_2+y)\geq f(0,x'_2-y)$ for every  $y\in (-\infty,  \gamma]$. 
Define, for $y=(y_1,y_2)\in \R^2$, the function $H_{x,y_1}(y_2)=f(x-y)g(y)$. When $y_1$ is clear from the context, we shall write $H_{x}(y_2)$ instead of $H_{x,y_1}(y_2)$ for simplicity of notation. Now write, beginning with the change of variable $y_2 \mapsto -y_2$,
\begin{linenomath*}\begin{align*}
    h(x)&=\int_{y_1\in \R}\int_{y_2\in \R}  H_x(-y_2)dy_1dy_2=\int_{y_1\in \R}\int_{y_2\in \R}  H_x(-y_2-\gamma)dy_1dy_2 \\
     &=\int_{y_1\in \R}\left(\int_{y_2\geq 0}  H_x(-y_2-\gamma)dy_2 + \int_{y_2\leq 0}  H_x(-y_2-\gamma)dy_2\right)dy_1\\
     &=\int_{y_1\in \R}\int_{y_2\geq 0}  H_x(-y_2-\gamma)+  H_x(y_2-\gamma)dy_2dy_1,
\end{align*}\end{linenomath*}
and
\begin{linenomath*}\begin{align*}
     h(x')&=\int_{y_1\in \R}\int_{y_2\in \R} H_{x'}(y_2)dy_1dy_2\\
    &=\int_{y_1\in \R}\int_{y_2\geq \gamma}  H_{x'}(y_2) + \int_{y_2\leq \gamma}  H_{x'}(y_2)dy_1dy_2\\
      &=\int_{y_1\in \R}\left(\int_{y_2\geq 0}  H_{x'}(y_2+\gamma)dy_2 + \int_{y_2\leq 0}  H_{x'}(y_2+\gamma)dy_2\right)dy_1\\
      &=\int_{y_1\in \R}\left(\int_{y_2\geq 0} H_{x'}(y_2+\gamma) +  H_{x'}(-y_2+\gamma)dy_2\right)dy_1
\end{align*}\end{linenomath*}
Hence, we have that $h(x)-h(x')$ is equal to
\begin{linenomath*}\begin{align*}
     \int_{y_1\in \R}\int_{y_2\geq 0} &f( -y_1,x_2+y_2+\gamma)g(y_1,-y_2-\gamma)+ f( -y_1,x_2-y_2+\gamma)g(y_1,y_2-\gamma)\\
     &- f( -y_1,x_2'-y_2-\gamma)g(y_1,y_2+\gamma)- f( -y_1,x_2'+y_2-\gamma)g(y_1,\gamma-y_2)dy_1dy_2
\end{align*}\end{linenomath*}
Since $g$ is radial, we have $g(y_1,-y_2-\gamma)=g(y_1,y_2+\gamma)$ and $g(y_1,\gamma-y_2)=g(y_1,y_2-\gamma)$. Furthermore, using that $x_2+\gamma=x'_2-\gamma$, we obtain that $h(x)-h(x')$ is equal to:
\begin{linenomath*}\begin{align*}
     \int_{y_1\in \R}\int_{y_2\geq 0}\left (f( -y_1,x_2+y_2+\gamma)-f( -y_1,x_2-y_2+\gamma)\right)\left(g(y_1,y_2+\gamma)-g(y_1,y_2-\gamma)\right)dy_1dy_2
\end{align*}\end{linenomath*}
In this summation, since $x_2\geq 0$, $\gamma\geq 0$ and $y_2\geq 0$, we have $\lvert x_2+y_2+\gamma  \rvert \geq \lvert x_2-y_2+\gamma\rvert$ and $\lvert y_2+\gamma \rvert \geq \lvert y_2-\gamma \rvert $. Since $f$ and $g$ are non-increasing functions of the distance to $0$, both factors of the integrand are non-negative, hence the integrand is non-negative and $h(x)-h(x')\geq 0$.
\end{proof}

\subsection{Projections of 2-dimensional L\'evy walks are also L\'evy} \label{appendix:projection-proof}

Consider a L\'evy walk $Z^\mu$ with parameter $\mu$ on $\R^2$, that has maximal step length $\ell_{max}$ (including the case $\ell_{max}=\infty$). 
It is well-known that the projection of a L\'evy walk with parameter $\mu$
     on each of the axes is also a L\'evy walk with parameter $\mu$. For example, 
     the conservation of the power-law distribution under projection was established by Sims et al.~[13].
     Nevertheless, in this section, we  provide another proof for this fact, for completeness purposes, and also because [13]
     did not examine the case $\lmax<\infty$.
     
     Without loss of generality, we may consider only the projection
    $Z_1^\mu$ on the $x$-axis. Hence, we aim to prove the following.
\begin{theorem}\label{thm:projection}
The projection
    $Z_1^{\mu}$ of $Z^{\mu}$ is a L\'evy walk on $\R$ with parameter $\mu$, in the sense that the p.d.f.~of the step-lengths of $X_1^{\mu}$ is $p(\ell)\sim 1/\ell^\mu$, for $\ell\in [1,\frac{\ell_{max}}{2}]$. Furthermore, the variance of $X_1^{\mu}$ is
\begin{linenomath*}\[
\sigma'^2= \begin{cases}
\Theta({\ell^{3-\mu}_{max}}) \text{ if } \mu\in(1,3) \\
  \Theta(\log \ell_{max}) \text{ if } \mu=3\end{cases}.
\] \end{linenomath*}
    \end{theorem}

\begin{proof}It is clear that $Z_1^\mu$ is also a random walk that moves incrementally, with the increments between $Z_1^\mu(m)$ and $Z_1^\mu(m+1)$ being the projection $Z_1(m+1)$ of the chosen 2-dimensional vector $V(m+1)=Z^\mu(m+1)-Z^\mu(m)$. These projections are i.i.d.~variables as the vectors $(V(m))_m$ are i.i.d.~variables, and their signs are $\pm$ with equal probability. Hence, all that needs to be verified is that $l_1:=\lvert V_1(1) \rvert $ has a L\'evy distribution with parameter $\mu$.

Let $V$ be one step-length drawn according to a L\'evy distribution $p^{\mu}$. Recall that
\begin{linenomath*}\[ p^{\mu}(\ell)=\begin{cases} a_{\mu} \text{ if } \ell \leq 1 \\
a_{\mu}\ell^{-\mu} \text{ if } \ell\in [1,\lmax) \\
0 \text{ if } \ell \geq \lmax
\end{cases}, \] \end{linenomath*}
where $a_{\mu}$ is the normalization factor, with $a_{\mu}=\frac{1}{1+\int_{1}^{\ell_{max}} \ell^{-\mu} d\ell }=\frac{1}{1+\frac{1-\ell_{max}^{1-\mu}}{\mu-1}}\in [1-\frac{1}{\mu},$.
Hence the distribution of $V=(V_1,V_2)\in \R^2$ is
\begin{linenomath*}\begin{equation} p^V(x)= \frac{1}{2\pi}\frac{1}{\norm{ x} } p^{\mu}(\norm{ x})=\begin{cases}\frac{a_{\mu}}{2\pi}\norm{ x}^{-1} \text{ if } \norm{ x}  \leq 1\\ \frac{a_{\mu}}{2\pi}\norm{ x} ^{-\mu-1}\text{ if } \norm{ x}\in [1,\ell_{max}) \\
0 \text{ if } \norm{x} \geq \lmax\end{cases}. \label{eq:law-step-R2}
\end{equation}\end{linenomath*}
For $x_1\in (0,\ell_{max})$, we have
\begin{linenomath*}\begin{align*}
     p^{l_1}(x_1)&=2\int_{0}^{\sqrt{\ell_{max}^2-x_1^2}} p^V(x_1,x_2)dx_2 \\
     &=\frac{2a_{\mu}}{2\pi} \int_{0}^{\sqrt{\ell_{max}^2-x_1^2}} \mathbf{1}_{\norm{x}< 1}\frac{1}{\norm{x}} + \mathbf{1}_{\norm{x}\geq 1}\frac{1}{\norm{x}^{1+\mu}}dx_2,
\end{align*}\end{linenomath*}
where $x=(x_1,x_2)$. If $|x_1|\geq 1$, then $\norm{x}\geq 1$ for any $x_2\in \R$, so that

\begin{linenomath*}\begin{align*}
     p^{l_1}(x_1) &=\frac{a_{\mu}}{\pi} \int_{0}^{\sqrt{\ell_{max}^2-x_1^2}} \frac{1}{(x_1^2+x_2^2)^{\frac{1+\mu}{2}}}dx_2 \\
     &=\frac{a_{\mu}}{\pi} \frac{1}{x_1^\mu} I(x_1),
\end{align*}\end{linenomath*}
where
\begin{linenomath*}\[ I(x_1):= \int_{0}^{\sqrt{\frac{\ell_{max}^2}{x_1^2}-1}} \frac{1}{(1+y^2)^{\frac{1+\mu}{2}}} dy. \] \end{linenomath*}
For any $x_1\in (1,\ell_{max})$, we have $I(x_1)\leq \int_{0}^{\infty} \frac{1}{(1+y^2)^{\frac{1+\mu}{2}}} dy=O(1)$ since $\frac{1}{(1+y^2)^{\frac{1+\mu}{2}}}=\Theta(y^{-\mu})$, for large $y$, and this function of $y$ is integrable as $\mu>1$. Furthermore, if $\lvert x_1\rvert \leq \ell_{max}/2$, we have $I(x_1)\geq \int_{0}^{1} \frac{1}{(1+y^2)^{\frac{1+\mu}{2}}} dy$ which is a positive constant. Hence, if $\lvert x_1 \rvert \in (1,\ell_{max}/2)$, we have
\begin{linenomath*}\begin{equation}\label{eq:projectmedium}
  p^{l_1}(x_1) = \Theta\left ( \frac{1}{x_1^\mu}\right),
\end{equation}\end{linenomath*}
and for $\ell_{max}/2\leq x_1\leq \ell_{max}$, we have 
\begin{linenomath*}\begin{equation}\label{eq:projectlarge}
p^{l_1}(x_1) = O\left(\frac{1}{x_1^\mu}\right).
\end{equation}\end{linenomath*}
Hence, the projection of the L\'evy walk on the axes are L\'evy-like, in the sense that their step-lengths distributions generally follow a power-law of same exponent~$\mu$.
The expected length, second moment and variance of one projected step are computed as in Claim \ref{claim:exp+var}. Indeed write, for $i\in \{1,2\}$,
\begin{linenomath*}\[ \int_0^{\lmax} x_1^i p^{l_1}(x_1)dx_1=\Theta\left(\int_0^1 x_1^i p^{l_1}(x_1)dx_1+\int_1^{\lmax/2} x_1^{i-\mu} dx_1 +\int_{\lmax/2}^{\lmax} x_1^ip^{l_1}(x_1) dx_1 \right).\] \end{linenomath*}
We have $\int_0^1 x_1^i p^{l_1}(x_1)dx_1 \leq 1$. Also, it is easy to verify from Eq.~\eqref{eq:projectmedium} and \eqref{eq:projectlarge} that the third term is dominated by the second term, which in turn, is $\Theta(\int_1^{\lmax} x_1^{i-\mu} dx_1)$. Hence, the expected length, second moment and variance of one projected step  are of the same order as those of the non-projected steps given by Claim \ref{claim:exp+var}, which concludes the proof of Theorem \ref{thm:projection}.
\end{proof}

\section{Lower Bounds}\label{sec:lower}
\subsection{Random walk with a fixed step-length}\label{sec:illustration}

In order to illustrate the definition of the overrun, we provide here a simple computation of the overrun of the intermittent process $X$ in which all step lengths are some pre-determined fixed integer $\ell$. Note that the case $\ell=1$ corresponds to the simple random walk, and that taking $\ell=\Theta(\sqrt{n})$ may be viewed as a ballistic strategy.
Consider a disc target of diameter $D<\sqrt{n}/2$. Since the searcher starts at a random point, with constant probability, the target is located at a distance of at least $\sqrt{n}/4$ from the initial location of the searcher. In this case, merely traversing this distance by the random walk process requires  $\Omega((\sqrt{n}/\ell)^2)=\Omega(n/\ell^2)$ steps on expectation, and hence consumes $\Omega(n/\ell)$ time on expectation. This implies that $\Over^X(n,D)=\Omega(D/{\ell})$.  Furthermore, as illustrated in the main text (Fig.~1b), and as shown formally in the next section, there are $\Omega(n/D^2)$ possible locations of the target. Since the agent must, on average, visit at least half of those, it will overall need $\Omega(n\ell/D^2)$ time to find the target on expectation, since each step takes  $\ell$ time. Thus, we also have $\Over^X(n,D)=\Omega(\ell/D)$.
Altogether, these arguments imply that $\Over^X(n,D)=\Omega(\max\{\ell/D, D/\ell\})$. 
While $\ell$ can be tuned to optimize the overrun with respect to a specific value of $D$, if we know only an upper bound $D_{max}$ on the value of $D$ then the overrun would be large with respect to either $D=1$ or $D=D_{max}$. Specifically, for $D=1$ we have $\Over^X(n,1)=\Omega(\ell)$, while for $D=D_{max}$, we have $\Over^X(n,D_{max})=\Omega(D_{max}/\ell)$. Hence, for at least one value of $D$ among the two, we have $\Over^X(n,D)=\Omega(\sqrt{D_{max}})$. In particular, if $D_{max}=n^\delta$ for some $\delta>0$ then the overrun is polynomial in $n$. 
\subsection{General lower bounds}\label{sec:generallowerbounds}
We prove here a general proposition that holds for any search process $X$ on the torus whose speed is constant (i.e., it takes $O(\ell)$ units of time to do a ballistic step of length $\ell$). We may assume without loss of generality that the speed is normalized to 1. Note also that, since we aim at a lower bound, we can suppose, without loss of generality, that the scan time in-between steps is $0$.

 We next define a quantity, termed $T_d$, which will be used to lower bound the time needed to detect an extended target $B(S)$ at distance $d$ or more. Formally, we distinguish between two cases, according to the given process $X$. 
 \begin{itemize}
     \item 
 If $X$ is an intermittent random walk, 
 we let $T_d$ be the expected time needed before the end point of a step is at distance at least $d$ from the initial location.
 \item Otherwise, we simply define $T_d=d$. 
  \end{itemize}
 
\begin{claim}\label{claim:lb-time-to-go-to-distance-D}
Let $X$ be any search process on the torus. Consider any target $S$ of diameter $D<\sqrt{n}/6-1$.  
The expected time to detect $S$ is  $\Omega (n \frac{T_{D}}{D^2})$.
\end{claim}
\begin{proof}
Consider a target $S$ of diameter $D$ and of an arbitrary shape. Instead of considering that $S$ is fixed and that the initial location $X(0)$ is chosen u.a.r, we may assume without loss of generality that $X(0)$ is fixed, say at the origin, and that the center of mass $u^\star$ of $S$ is chosen 
uniformly at random in the torus. 

Let us first construct a grid with  $s\times s$ nodes, where $s=\lfloor \sqrt{n}/(3D+2)\rfloor$. Note that since $D<\sqrt{n}/6-1$, we have $s\geq 2$. To make the grid symmetric, we let the distance between two neighboring nodes be precisely $\sqrt{n}/s$.
We next align the grid so that $u^\star$ is a node of the grid, and construct a disc of radius $D+1$ around each node. Note that the number of discs is $M=s^2=\Omega(n/D^2)$, and that the distance between any two discs is at least $D$. See Figure 1(b) in the main text. 
Furthermore, note that the disc $U^\star$ corresponding to $u^\star$ fully contains the extended target $B(S)$. Let us therefore lower bound the time until visiting $U^\star$ for the first time. This will serve as the desired lower bound for detecting~$S$.

Assume that the information about the collection of discs is given to the searcher. We may assume this, since it can only decrease the best detection time.
Because the location of $S$ in chosen u.a.r in the torus, from the perspective of the searcher, each of the discs has an equal probability to be $U^\star$. 
It follows that with probability~$1/2$, at least half of the discs are visited, before the searcher visits $U^\star$. Since the discs are separated by distance of at least $D$, we immediately get that the expected time until visiting $U^\star$ is $\Omega(MD)=\Omega(n/D)$, which is the desired claim when $X$ is not an intermittent random walk (and hence $T_D=D)$. 

Let us next consider the case that $X$ is an intermittent random walk. The arguments are similar, yet slightly more subtle. 
We aim to lower bound the time until visiting $U^\star$ for the first time, where by visiting a disc, we mean that the end of a ballistic step of $X$ is in that disc. For this purpose, we may assume that the process terminates when it visits 
$U^\star$. Let $U_1,U_2,\ldots$ denote the newly visited discs, in order of visitation, with all the $U_i$ distinct. Let $A_i$ be the event that $U^\star\notin \{U_1,\ldots, U_i\}$. Note that $\Pr(A_i)=1-\frac{i}{M}$.
Let $t_i$ denote the time from visiting $U_i$ (for the first time) until visiting $U_{i+1}$ (for the first time), in the event that $A_i$ occurs. If the event $A_i$ does not occur, we say that $t_i=0$. 
 The time before visiting $U^\star$ can therefore be written as $\sum_{i=1}^{M-1} t_i$. Furthermore, we have $\E(t_i)=\E(t_i\mid A_i)\Pr(A_i)$. Hence, the expected time before visiting $U^\star$ is:
\begin{linenomath*}\[
\sum_{i=1}^{M-1} \E(t_i\mid A_i) \Pr(A_i).
\]\end{linenomath*}
Now recall that $X$ is an intermittent Markovian process, and that $A_i$ corresponds to an event that is relevant up to (and including) the detection of $U_i$. Hence,  
$\E(t_i\mid A_i)$ is lower bounded by the minimal expected time that the intermittent random walk $X$, starting at some point $u\in U_i$, visits another disc, where the minimization is taken w.r.t $u\in U_i$. Since discs are separated by distance of at least $D$, 
the process starting at any such $u$ needs to visit a disc at distance at least $D$.
It therefore follows that $\E(t_i\mid A_i)\geq T_D$. Altogether, the expected time to detect $S$ is at least: \begin{linenomath*}\[
\sum_{i=1}^{M-1} T_D \Pr(A_i)=
\sum_{i=1}^{M-1} T_D (1-i/M)=
\Omega(T_D M)= \Omega\left(n\frac{T_D}{D^2}\right),
\]\end{linenomath*}
as desired.
\end{proof}
\begin{corollary}\label{cor:obs:lower}
For every $1\leq D\leq \sqrt{n}/2$, the best possible detection time is $\Theta(n/D)$, when we allow the strategy to have continuous detection, to be unrestricted in terms of its internal computational power and navigation abilities, and to be fully tuned to the diameter. In other words, $\opt(n,D)=\Theta(n/D)$.
\end{corollary}
\begin{proof}
The fact that $\opt(n,D)=\Omega(n/D)$ for every $D<\sqrt{n}/6-1$ follows immediately from Claim~\ref{claim:lb-time-to-go-to-distance-D} and the fact that $T_D\geq D$. For $\sqrt{n}/6-1<D\leq \sqrt{n}/2$ the bound 
$\Omega(n/D)=\Omega(\sqrt{n})$  follows simply because with constant probability, the target is at distance $\Omega(\sqrt{n})$ from the initial location of the searcher. 

In order to see why $\opt(n,D)=O(n/D)$, let us tile the torus with horizontal and vertical lines partitioning the torus into squares of size $D/2\times D/2$ each. 
In the case that $\sqrt{n}$ is not a multiple of $D/2$, we might have few of these squares smaller than $D/2\times D/2$. It is clear that this can be constructed while maintaining that the number of horizontal and vertical lines is $O(\sqrt{n}/D)$. For any connected target $S$ of diameter $D$, the set $B(S)$ must intersect at least one of these lines. Now consider a deterministic strategy that repeatedly walks over this tiling exhaustively, without doing much repetition in each exhaustive search. E.g., by first walking on the horizontal lines exhaustively (with occasional steps to move between horizontal lines) and then 
walking on the vertical lines exhaustively.
It is easy to see that such a strategy exists and requires at most $O(\sqrt{n}/D\cdot \sqrt{n})=O(n/D)$ time to pass over all the lines, and hence to detect the target. This establishes the required upper bound.
\end{proof}

Claim \ref{claim:lb-time-to-go-to-distance-D}, applied with $D=1$, also yields the following corollary, by remarking that for intermittent random walk processes, $T_D$, namely,
the expected time until the end point of a step is at a distance of at least $D$ is at least the expected time for one step $\tau$, i.e., $T_D\geq \tau$.

\begin{corollary}\label{cor:scale-sensitivity-expected-length}
Consider an intermittent random walk strategy $X$ on the torus $\T_n$. The detection time of any  target of diameter $D$ is $\Omega(n\tau/D^2)$. 
\end{corollary}

\begin{claim}\label{claim:max-length}
Consider a random walk process $X$ on the torus $\T_n$ and let $\sigma'$ denote the standard deviation of the length of the projected steps onto either coordinate. 
\begin{itemize}
    \item The expected maximal distance of $X$ to its origin after $m$ steps, i.e. $\max_{s\leq m} \norm{X(s)-X(0)}$, is $O(\sqrt{m}\sigma')$. 
    \item Let $m_d$ be the number of steps needed to go to distance at least $d<\sqrt{n}/2$, in other words $m_d$ is the first step $m$ for which $\norm{X(m)-X(0)}\geq d$. We have $\E(m_d)=\Omega(d^2/\sigma'^2)$. 
    \item If the process is intermittent and $\tau$ denotes the average length of a jump, then the expected time before reaching distance $d<\sqrt{n}/2$ is $T_d=\Omega(d^2/\sigma'^2 \tau)$.
\end{itemize}
In particular, if the process is intermittent and $L$ is the maximal length in the support of the step-length distribution, then the expected time needed to go to a distance $\Omega(\sqrt{n})$ is $\Omega( \frac{n}{L})$.
\end{claim}
We will use Claim \ref{claim:max-length} in the next section to get an upper bound on the time needed for a L\'evy walk to reach some distance. The proof of Claim \ref{claim:max-length} is based on Kolmogorov's inequality.
\begin{proof}
 Let $Z$ be the process on $\R^2$, with $Z(0)=X(0)$ and evolving with the same steps as $X$. Since the distance between $Z(m)$ and $Z(0)$, in $\R^2$, is always at least that of $X(m)$ and $X(0)$, in $\T_n$, the number of steps needed to go to distance $d$ in $\T_n$ is at least as high as in $\R^2$. Hence, we may analyze the process $Z$ instead of $X$.
 
 Define $d^Z_{max}(m)$ as the maximal distance (from the initial point) that the process $Z$ reached from step~$0$ up to step $m$, i.e., \begin{linenomath*}\[d^Z_{max}(m)=\max_{s\leq m} \norm{Z(0)- Z (s)}.\] \end{linenomath*}
	Now write $Z=(Z_1,Z_2)$, let $p'$ be the p.d.f. of the projected step-lengths (i.e. the p.d.f. of the step-lengths of $Z_i$), and let $\tau'$ and $\sigma'$ be respectively its mean and standard deviation. 
	Next, let $d^Z_{i,max}(m)$ be the maximal distance reached by the projection on coordinate $i=1,2$. Since steps are independent, the standard deviation of $Z_i(s)$, for $s\leq m$, is $\sqrt{s}\sigma'\leq \sqrt{m}\sigma'$.
	
	By Kolmogorov's inequality, we have for any $\lambda>0$, $\Pr(d^Z_{i,max}(m) \geq \lambda \sqrt{m} \sigma' )\leq \frac{1}{\lambda^2}$. Furthermore, since $d^Z_{max}(m)\leq \sqrt{2}\max\{d^Z_{1,max}(m),d^Z_{2,max}(m)\}$, we have by a union bound argument, for any $\lambda>0$, 
	\begin{linenomath*}\begin{equation}
    \Pr(d^Z_{max}(m)\geq \lambda\sqrt{m}\sigma') \leq  \Pr\left (d^Z_{1,max}(m)\geq \frac{\lambda}{\sqrt{2}}\sqrt{m}\sigma' \right ) + \Pr\left( d^Z_{2,max}(m)\geq \frac{\lambda}{\sqrt{2}}\sqrt{m}\sigma'\right) \leq \frac{4}{\lambda^2}.\end{equation}\end{linenomath*}
	Hence, 
	\begin{linenomath*}\begin{align} \E(d^Z_{max}(m))&=\int_{s=0}^\infty \Pr\left(d^Z_{max}(m)\geq s\right)ds \leq \sum_{\lambda=0}^\infty \int_{\lambda'=0}^{\sqrt{m}\sigma'} \Pr\left(d^Z_{max}(m)\geq \lambda \sqrt{m}\sigma' + \lambda'\right) d\lambda' \nonumber \\
	&\leq \sqrt{m}\sigma' \left(\sum_{\lambda\geq 0} \Pr(d^Z_{max}(m)\geq \lambda \sqrt{m}\sigma')\right)=O\left (\sqrt{m}\sigma'\right), \label{eq:max-distance-general-upper-bound}\end{align}\end{linenomath*}
	which proves the first item of Claim \ref{claim:max-length}. Next, write the $m_d$ of the statement as $m^X_d$, to distinguish it from the similarly defined $m^Z_d$, which is the first step for which $\norm{Z(m)-Z(0)}\geq d$. As remarked above, we have $m^X_d\geq m^Z_d$. Note that for $m\geq m^Z_d$,  we have $d^Z_{max}(m)\geq d^Z_{max}(m_d)\geq d$. Therefore, by Markov's inequality,
\begin{linenomath*}\begin{equation} \label{eq:markov-inq-distances-steps}
\E(d^Z_{max}(2\E(m^Z_d)))\geq\E(d^Z_{max}(2\E(m^Z_d))\mid m^Z_d<2\E(m^Z_d))\cdot \Pr(m^Z_d<2\E(m^Z_d))
\geq d \cdot \frac{1}{2}. \end{equation}\end{linenomath*}
Now using Eq.~\eqref{eq:max-distance-general-upper-bound} with $m=2\E(m^Z_d)$, we have $\E(d^Z_{max}(2\E(m^Z_d)))=O(\sqrt{\E(m^Z_d)}\sigma')$ and hence, by Eq.~\eqref{eq:markov-inq-distances-steps},
\begin{linenomath*}\[ \E(m^X_d)\geq \E(m^Z_d)=\Omega\left(\frac{d^2}{\sigma'^2}\right), \]\end{linenomath*}
which proves the second item of Claim \ref{claim:max-length}. 

The last item is a lower bound on $T_d=\E(T(m^X_d))$, the expected time that $X$ needs to reach distance $d$. To obtain it, we observe that $m^X_d$ is the hitting step of the set of nodes at distance $d$ or more in the torus. Hence, by Claim \ref{claim:timeVSmoves}, we have $T_d=\E(m^X_d)\cdot \tau\geq \E(m^Z_d)\cdot \tau=\Omega(\frac{d^2}{\sigma'^2}\tau)$, which was exactly as needed.

Finally, observe that
\begin{linenomath*} \begin{equation}\label{eq:variance1}
     \sigma'^2 = \int_{0}^{L} p'(\ell)\ell^2d\ell \leq \int_{0}^{L} p'(\ell)\ell \cdot L d\ell = L\tau'\leq L\tau,
     \end{equation}\end{linenomath*}
     where the last inequality is justified by the fact that the projection reduces distances. This completes the proof of Claim \ref{claim:max-length}.
\end{proof}

\subsection{Lower bounds for L\'evy walks}\label{sec:lowerlevy}
The goal of this section is to prove lower bounds on the overrun of L\'evy walks other than Cauchy. For $1<\mu< 2$, we show that the corresponding intermittent L\'evy walks are bad at finding small targets. For $2<\mu\leq 3$, we show that the corresponding  L\'evy walks are bad at finding large targets. The latter result holds also with respect to the continuous detection model.

\subsubsection{Intermittent L\'evy walks with \texorpdfstring{$1<\mu\leq 2$}{small exponent}}
Let $X^\mu$ be the intermittent L\'evy walk on the torus $\T_n$, for some $1<\mu< 2$. 
 We start by analyzing the detection times of small targets.
 
 \begin{theorem}\label{ThmLevyLowerSmallMu}
Let $\mu \in (1,2)$ and $D\in [1,\sqrt{n}/2]$. Write $\mu=2-\epsilon$. The detection time of the L\'evy walk $X^\mu$ with respect to a target $S$ of diameter $D$ is
\begin{equation}
    t_{detect}^{X^\mu}(S)=\Omega(n^{1+\epsilon/2}/D^2), \end{equation}
and the overrun  w.r.t. $D$ is:
\begin{linenomath*}\begin{equation}
    \Over^{X^{\mu}}(n,D)= \Omega(n^{\epsilon/2}/D).
    \end{equation}\end{linenomath*}
\end{theorem}
    \begin{proof}
    By Corollary \ref{cor:scale-sensitivity-expected-length}, the detection time of a target $S$ with diameter $D$ is $\Omega(n\tau/D^2)$ where $\tau$ is the expected step length. 
    Using that $\lmax=\Theta(\sqrt{n})$, Claim \ref{claim:exp+var} implies that this expected step length is, for $\mu=2-\epsilon$ with $\epsilon\in (0,1)$:
\begin{linenomath*}\[\tau = \Theta(n^{1-\frac{\mu}{2}})=\Theta(n^{\eps/2}).\]\end{linenomath*}
Hence, the detection time $X^\mu$ for a target of diameter $D$ is $\Omega(n^{1+\eps/2}/D^2)$.
Dividing this by the unconditional optimal time $\Theta(n/D)$, we get the desired lower bound on the overrun.
\end{proof}

\subsubsection{L\'evy walks with \texorpdfstring{$2<\mu\leq 3$}{large exponent}}
Theorem \ref{ThmLevyLowerSmallMu} implies that the overrun of the intermittent L\'evy walk $X^\mu$ for $\mu\in(1,2)$ is very large with respect to small targets, i.e, when $D\ll n^{\epsilon/2}$. We next aim to prove the case $\mu\in [2,3]$:
\begin{theorem}\label{ThmLevyLowerLargeMu}
Let $\mu\in(2,3]$ and $D\in [2,\sqrt{n}/6-1]$. Write $\mu=2+\epsilon$ where $0< \epsilon\leq 1$. 
The following holds with respect to the L\'evy process $X^\mu$ whether it is intermittent or not.
The detection time of $X^\mu$ with respect a target $S$ of diameter $D$ is
\begin{linenomath*}\[t_{detect}^{X^\mu}(S)=\begin{cases}
 \Omega({n}{D^{\eps-1}}) \text{ if } \mu=2+\epsilon, \mbox{~where~} 0<\epsilon<1, \\
 \Omega(\frac{n}{\log D})  \text{ if } \mu=3. 
\end{cases}\]\end{linenomath*}
Hence, the overrun of $X^\mu$ with respect to $D$ is:
\begin{linenomath*}\[\Over^{X^{\mu}}(n,D)=\begin{cases}
 \Omega(D^{\eps}) \text{ if } \mu=2+\epsilon, \mbox{~where~} 0<\epsilon<1, \\
 \Omega(\frac{D}{\log D})  \text{ if } \mu=3.
\end{cases}\]\end{linenomath*}
\end{theorem}
Since the proof is simpler, let us first prove Theorem \ref{ThmLevyLowerLargeMu} for the intermittent setting, i.e., targets can only be detected in-between steps. 

\paragraph{Proof of Theorem \ref{ThmLevyLowerLargeMu} for the intermittent setting.}
Towards proving the theorem, we first establish the following. 
\begin{claim}\label{claim:distance_large_mu}
 Let $X^{\mu}$ be an intermittent L\'evy walk process on the torus $\T_n$, for $\mu\in [2,3]$, with $\ell_{max}=\sqrt{n}/2$. The expected time required to reach a distance of 
 $d\geq 1$ from the starting point is:
\begin{linenomath*} \[T_d=
 \begin{cases}
  \Omega(d\log d) \text{ if } \mu=2 \\
\Omega(d^{\mu-1}) \text{ if } \mu\in(2,3) \\
  \Omega(\frac{d^2}{\log d}) \text{ if } \mu=3\end{cases}.
 \]\end{linenomath*}
\end{claim}
\begin{proof}
We may suppose that $d\in[1,\sqrt{n}/4]$. Denote by $m_d$ the random number of \textit{steps} before the process reaches a distance of at least~$d$. Let us define $m_0=\lceil d^{\mu-1} \rceil$, and say that a step is {\em small} if it has length at most $d$. Define the event $\mathcal{A}$ that all the steps $1,2,\dots,m_0$ are small. Note that since $d\leq \ell_{max}/2$, the probability for any given step not to be small is $q=\int_{d}^{\lmax} \frac{a}{\ell^{\mu}}d\ell\geq \frac{c}{d^{\mu-1}}$ for some constant $c\in (0,1)$. Hence, the probability for a step to be small is $1-q$, and since the steps are independent, we have:
\begin{linenomath*}\[\Pr(\mathcal{A})=(1-q)^{m_0}=\exp(m_0 \log(1-q))\geq \exp(d^{\mu-1}\log(1-cd^{1-\mu})).\]\end{linenomath*} We have: \begin{linenomath*}\[\exp(d^{\mu-1}\log(1-cd^{1-\mu}))=\exp(d^{\mu-1}(-cd^{1-\mu}+o(d^{1-\mu}))=\exp(-c+o(1)),\]\end{linenomath*} which is a positive constant. Since this is a continuous, strictly positive, function of $d\in[1,\infty)$, we have $\Pr(\mathcal{A})\geq c'$ for some constant $c'>0$ independent of $d$.
 
 Next, note that \begin{linenomath*}\[\E(T(m_d))\geq \Pr(\mathcal{A}) \cdot\E(m_d\mid \mathcal{A})=c'\cdot \E(T(m_d)\mid\mathcal{A}).\]\end{linenomath*}
 Hence, for the purposes of obtaining a lower bound, it is sufficient to examine the process when conditioned on $\mathcal{A}$. This is a L\'evy process of parameter $\mu$, with cut-off $\lmax=d$. The expected length $\tau$ of a jump is given by Claim~\ref{claim:exp+var}: \begin{linenomath*}\begin{equation}
\tau=\Theta(1)
\end{equation}\end{linenomath*} and the variance $\sigma'^2$ of the step-length of a jump projected onto one of the axes is given by Theorem~\ref{thm:projection}:
 \begin{linenomath*}\[
\sigma'^2= \begin{cases}
\Theta({d^{3-\mu}}) \text{ if } \mu\in (1,3) \\
  \Theta(\log d) \text{ if } \mu=3\end{cases}.
\]\end{linenomath*}
 To conclude, we use Claim \ref{claim:max-length}:
\begin{linenomath*}\[ T_d=\Omega \left( \frac{d^2}{\sigma'^2}\cdot \tau\right)=\begin{cases}
\Omega(d^{\mu-1}) \text{ if } \mu\in(2,3) \\
  \Omega(\frac{d^2}{\log d}) \text{ if } \mu=3\end{cases}.\]\end{linenomath*}
  This concludes the proof of Claim \ref{claim:distance_large_mu}.
  \end{proof}
  
  Combining Claim \ref{claim:distance_large_mu} with the fact
  that the expected time to detect a target of diameter $D$ is  $\Omega (n \frac{T_{D}}{D^2})$, as established by Claim \ref{claim:lb-time-to-go-to-distance-D}, and comparing to the unconditional optimal detection time $\Theta(n/D)$ for targets of diameter $D$, Theorem \ref{ThmLevyLowerLargeMu} is proved in the intermittent case. Next, we prove the theorem when the process is able to detect the target while moving.
  
  
  \paragraph{Proof of Theorem \ref{ThmLevyLowerLargeMu} for the continuous detection model.} Recall, from the proof of Claim~\ref{claim:lb-time-to-go-to-distance-D}, that we can build a grid of $M=\Theta(n/D^2)$ discs of diameter $D$, one of which contains the target, and separated by distance $D$. Furthermore, for every strategy, whether intermittent or not, with probability $\frac{1}{2}$, at least half of the discs are visited before finding the target. Hence, the expected time to find the target is at least half of the expected time to visit half of the discs.
In the remaining of the proof we aim to lower bound the expected time to visit half of the discs.
  
  Let $\mu>2$ and write $\mu=2+\epsilon$. Define a step to be {\em large} if it has length $D$ or more. Divide the execution into  a sequence of consecutive {\em phases}, so that each phase is a succession of small steps, and a final large step (possibly, there are no small steps in the phase if two large steps are consecutive). In short, in what follows we prove that a phase visits $O(1)$ discs on average when $2<\mu<3$, or $O(\log D)$ for $\mu=3$ (Lemma~\ref{lem:constant_disc}), and lasts, on average, $\Omega(D^{\mu-1})$ time (Lemma \ref{lem:time-phase}). We then conclude that, after $R=\tilde{\Theta}(M)$  phases, with constant probability, no more than $M/2$ discs are visited and the time spent is \begin{linenomath*}\[\tilde{\Omega}(MD^{\mu-1})=\tilde{\Omega}(nD^{\mu-3})=\tilde{\Omega}(nD^{\epsilon-1}).\]\end{linenomath*}
  A straightforward computation then allows to establish the desired bound on the overrun of the L\'evy search in the continuous detection model. 
  
  We next proceed to explain the proof in details. 
 Let $N_{discs}$ be the number of discs visited during a phase.
\begin{lemma}\label{lem:constant_disc}  
$\E(N_{discs})=\begin{cases}
O(1) \text{ if } 2<\mu<3\\
O(\log D) \text{ if } \mu=3 \end{cases}$.
\end{lemma}
\noindent{\em Proof of Lemma \ref{lem:constant_disc}.}
Given a phase, by linearity of expectation, $\E(N_{discs})$ equals the expected number of discs visited by the small steps of the phase plus the expected number of discs visited by the large step. The latter quantity is easy to bound. Indeed, since discs are separated by a distance of $D$, the number of discs visited in a step of length $L$ is $O(1+L/D)$. Moreover, it is easy to verify that, as $\mu>2$, the expected length of a large step is $\Theta(D)$. Hence the expected number of discs visited during the large step of a phase is $O(1)$.  

In the remaining of the proof of Lemma \ref{lem:constant_disc}, we aim to upper bound the expected number of discs visited by the small steps of the phase.



Let $D_{small}$ denote the number of discs discovered during the small steps. 
Towards establishing an upper bound on $\E(D_{small})$, let $\alpha$ be the probability for one step to be large. This equals $a\int_{\ell=D}^{\lmax} \ell^{-\mu}d\ell=\frac{a}{\mu-1}(D^{1-\mu}-\lmax^{1-\mu})$, and so, as $D<\lmax/2=\sqrt{n}/4$, we have:
\begin{linenomath*}\[\alpha=\Theta(D^{1-\mu}).\]\end{linenomath*}  Let $N_{small}$ be the total number of small steps in one phase. Since a phase ends after performing a long step for the first time, we have, for every integer $m\geq 0$, $\Pr(N_{small}=m)= \alpha(1-\alpha)^m$.
We thus have:
\begin{linenomath*}\begin{align} 
\E(D_{small})= \sum_{m\geq 0} \alpha (1-\alpha)^m \cdot \E(D_{small}\mid N_{small}=m).
\label{eq:Dsmall}\end{align}\end{linenomath*}
\begin{claim}\label{claim:Dsmall}
For any integer $m$, $\E(D_{small}\mid N_{small}=m)=O(1+m\sigma''^2/D^2)$, where $\sigma''$ is the standard deviation of the length of a small step, when projected on one of the coordinates.
\end{claim}
Note that the direction of each step is chosen uniformly at random, hence $\sigma''$ does not depend on which coordinate is chosen.

\noindent{\em Proof of Claim \ref{claim:Dsmall}.}
Let $W_1$ be the number of steps before a distance of $2D$ from the initial location is first reached. For $r\geq 1$, define recursively both $S_r=\sum_{i=1}^{r}W_i$, and $W_{r+1}$ to be the number of steps before we first have $\norm{X(S_r+W_{r+1})-X(S_r)}\geq 2D$. Note that the $(W_i)_i$ are i.i.d and have the same law as $m_{2D}$. Hence, by Claim \ref{claim:max-length},
we have \begin{linenomath*}\begin{equation}
\E(W_i)=\Omega(D^2/\sigma''^2).
\end{equation}\end{linenomath*}

For a given $m\geq 1$, let $r(m)$ be the first $r\geq 1$ for which $S_r> m$ (if this never happens then $r(m)=0$). Because in-between steps $W_i$ and $W_{i+1}$ only a distance $O(D)$ is travelled, there can only be $O(1)$ discs visited during this time  interval. Hence, up to step $m$, at most a number $O(1+r(m))$ discs are visited. We are thus looking for an upper bound on $\E(r(m))$.

Observe that $r(m)$ is a stopping time for the $(W_i)_{i\geq 1}$. Furthermore, $r(m)\leq m$ since $W_i\geq 1$ for all~$i$. Since the $W_i$ are i.i.d., and $\E(W_1)$ is finite also, we can apply  Wald's equation  
to obtain $\E(r(m))\E(W_1)= \E(S_{r(m)})$, and hence:
\begin{linenomath*}\begin{equation}
    \E(r(m))= \frac{\E(S_{r(m)})}{\E(W_1)}.
    \end{equation}\end{linenomath*}
Moreover, we have $\E(S_{r(m)})=\E(S_{r(m)-1})+\E(W_{r(m)})$. By definition of $r(m)$, we have $\E(S_{r(m)-1})\leq m$.
Next, we wish to bound $\E(W_{r(m)})$.
Note that $W_{r(m)}$ is at most the first $r>m$ for which $\norm{X(r)-X(m)}\geq 4D$. Indeed,  by definition of $r(m)$ we have $\norm{X(m)-X(S_{r(m)-1})}\leq 2D$ and $\norm{X(S_{r(m)})-S_{r(m)-1}}\geq 2D$. Hence we have $\E(W_{r(m)})\leq \E(m_{4D})+m$. Furthermore, we claim that $\E(m_{4D})=O(\E(m_{2D}))$. Indeed, consider a circle of radius $4D$ from the initial location and a step $s$, for which the agent is within the circle.
Consider $E=\E(S_3)=3\E(m_{2D})$. Starting at step $s$, with constant probability, there exists three steps $s_1,s_2,s_3\in(s,s+2E)$ for which $\norm{X(s_i)-X(s_{i+1})}\geq 2D$. Furthermore, whenever this happens, a distance of at least $4D$ from the center of the circle will be reached if $X(s_1),X(s_2)$ and $X(s_3)$ are aligned approximately in the direction leading to the shortest exit from the circle, which happens with constant probability. Hence, after $2E$ steps from any step $s$ where the agent is within the circle, with constant ability, the walk escapes the circle.  Applying this argument repeatedly implies that, $\E(m_{4D})=O(E)=O(\E(m_{2D}))$.
Altogether, we deduce that \begin{linenomath*}\[\E(r(m))=O(1+m/\E(m_D))=O(1+m\sigma''^2/D^2).\]\end{linenomath*}
As remarked above, up to step $m$, there are at most $O(1+ r(m))$ visited discs. Hence, conditioning on $N_{small}=m$, there are only $O(1+m\sigma''^2/D^2)$ discs visited in the small steps phase, on expectation.
This completes the proof of Claim \ref{claim:Dsmall}. \qed

 Using Claim \ref{claim:Dsmall}, we return to Eq.~\eqref{eq:Dsmall}, to bound the  expected number of discs visited in a small phase:
\begin{linenomath*}\begin{equation*}\label{eq:exp-D-small} 
\E(D_{small})=\sum_{m\geq 0} \alpha (1-\alpha)^m \cdot O(1+m\sigma''^2 /D^2)=
O(1)+O(\sigma''^2 \alpha/D^2 \cdot \alpha^{-2}),
\end{equation*}\end{linenomath*}
     where we used that $\sum_{m\geq 0} (1-\alpha)^m =\alpha^{-1}$, and that $\sum_{m\geq 0} m(1-\alpha)^m =O(\alpha^{-2})$.
     Thus, 
     \begin{linenomath*}\begin{equation}\E(D_{small})=O( 1+\sigma''^2 \alpha^{-1}/D^2).
\end{equation}\end{linenomath*}
As $\sigma''^2$ is the variance of the projected L\'evy distribution with cut-off $\lmax=D$, it is given by Theorem~\ref{thm:projection} as: $O(D^{3-\mu})$ for $\mu<3$ and $O(\log D)$ for $\mu=3$. Together with the fact that $\alpha=\Theta(D^{1-\mu})$, we get that the expected number of discs visited by the small steps of a phase is $O(1)$ for $\mu\in(2,3)$ and $O(\log D)$ for $\mu=3$. Combining with the expected number of discs visited by the large step, which was shown to be $O(1)$, the proof of Lemma \ref{lem:constant_disc} is complete. \qed

Given a constant $\tilde{c}$, define the following quantity that will refer to the number of phases. 
\begin{linenomath*}\begin{equation}\label{eq:R}
     R=\begin{cases}\tilde{c}M \mbox{ if } \mu\in (2,3)\\ \tilde{c}M/\log D \mbox{ if }  \mu=3\end{cases}.
     \end{equation}\end{linenomath*}
     Given $R$, let $N^R_{discs}$ denote the total number of discs visited by the end of the $R$-th phase. 
 \begin{lemma}\label{lem:number-discs-R-phases} For any $\delta<1$, there exists a constant $\tilde{c}>0$ such that the probability to have visited at most $M/2$ discs after $R$ phases (as defined in Eq. \eqref{eq:R}) is
  \begin{linenomath*}\[\Pr(N^R_{discs}< M/2) > \delta .\]\end{linenomath*}
  \end{lemma}
  \noindent{\em Proof of Lemma \ref{lem:number-discs-R-phases}.}
   Note that steps are independent and, hence, phases are independent, implying that the number of discs visited during a phase does not depend on the phase number.  We have, by linearity of expectation, $\E(N^R_{discs})=R\cdot\E(N_{discs})$, and, by Markov's inequality, we have \begin{linenomath*}\[\Pr(N^R_{discs}\geq M/2) \leq 2R\cdot \E(N_{discs})/M.\]\end{linenomath*} By Lemma \ref{lem:constant_disc},  $\E(N_{discs})\leq c$ for $\mu\in (2,3)$ and $\E(N_{discs})\leq c\log D$ for $\mu=3$, for some constant $c>0$. Hence,  we find that $2R\cdot \E(N_{discs})/M$ is at most $2c\tilde{c}$, which can be made  to be less than  $1-\delta$ by choosing $\tilde{c}< (1-\delta)/(2c)$.\qed

\begin{lemma}\label{lem:time-phase}
Let $T^R$ be the time spent during  $R$ phases. 
There are two constants $c>0$ and $q>0$ for which \begin{linenomath*}\[\Pr(T^R \geq c D^{\mu-1} R)\geq q.\]\end{linenomath*}
\end{lemma}
\noindent{\em Proof of Lemma \ref{lem:time-phase}.}
Define a phase to be long if it lasts at least $T^\star=c_1D^{\mu-1}$ time for some constant $c_1$ to be fixed later.
Let $N^R_{long-phases}$ be the number of long phases, up to the $R$-th one. 
Note that \begin{linenomath*}\begin{equation}\label{eq:TR}
T^R\geq T^\star N^R_{long-phases}.
    \end{equation}\end{linenomath*}
Let $T$ be the time duration of the small steps in a phase.
Since phases are independent, we have:
\begin{linenomath*}\begin{align}\E(N^R_{long-phases})&=R \cdot \Pr(T \geq T^\star)\geq R\cdot \Pr(N_{\geq \frac 12} \geq 2T^\star),
\end{align} \end{linenomath*}
where $N_{\geq \frac 12}$ is the number of steps of length larger than $\frac 12$ among the small steps of a phase. Because $N_{small}$, the number of small steps in one phase, follows a geometric distribution of parameter $\alpha^{-1}$, we have $N_{small}=\Omega(\alpha^{-1})$ with constant probability. Furthermore, as a small step has length at least $\frac 12$ with constant probability, we have that \begin{linenomath*}\[N_{\geq \frac 12}=\Theta(N_{small})=\Omega(\alpha^{-1}),\]\end{linenomath*} with constant probability. Indeed,  $N_{\geq\frac 12}$ follows a binomial distribution, and we are using the median property of such distributions.

By choosing $c_1$ such that $T^\star=c_1D^{\mu-1}$ is small enough, since $\alpha^{-1}=\Theta(D^{\mu-1})$, we have $\Pr(N_{\geq \frac 12} \geq T^\star)\geq c'$ for some constant $0<c'<1$.
This implies that for some constant $0<c''<1$, \begin{linenomath*}\[\E(N^R_{long-phases})\geq c''R.\]\end{linenomath*} Hence, 
\begin{linenomath*}\[\E(N^R_{short-phases})\leq (1-c'')R,\]\end{linenomath*} where $N^R_{short-phases}$ is the number of short (i.e., non-long) phases. By Markov's inequality, for any $c_2>0$, 
we have
$\Pr(N^R_{short-phases}\geq c_2R)\leq \frac{1-c''}{c_2}$, which is a positive, strictly less than $1$, constant, by a suitable choice of $c_2$. For this choice, we have \begin{linenomath*}\[\Pr(N^R_{long-phases}\geq (1-c_2)R)=\Pr(N^R_{short-phases}< c_2R)\geq 1-\frac{1-c''}{c_2}=\Omega(1).\] \end{linenomath*}
Returning to Eq.~\eqref{eq:TR}, we get that with constant probability
\begin{linenomath*}\[
T^R=\Omega(T^\star R)=\Omega(RD^{\mu-1}),
\]\end{linenomath*}
which proves Lemma \ref{lem:time-phase}.
\qed

We conclude by using Lemmas \ref{lem:number-discs-R-phases} and \ref{lem:time-phase}. Specifically,  for the constants $c>0$ and $0<q<1$ of Lemma~\ref{lem:time-phase}, 
and the constant $\delta=1-q/2$ in Lemma \ref{lem:number-discs-R-phases}, for some choice of the constant $\tilde{c}>0$ in the definition of $R$, we obtain:
\begin{itemize}
\item $\Pr(T^R \geq cR D^{\mu-1})>q$, and
\item $\Pr(N^R_{discs}< M/2)>\delta$.
\end{itemize}
Using a union bound argument, this implies that with probability at least $q+\delta-1=q/2$,  we have both $N^R_{discs}< M/2$ and $T^R=\Omega(RD^{\mu-1})$. Hence, with constant probability, the searcher takes time $\Omega(RD^{\mu-1})$ to find the target. Therefore, the expected time needed to find the target is \[t_{detect}^{X^\mu}(S)=\Omega(RD^{\mu-1})=\begin{cases}
\Omega (M D^{\mu-1})= \Omega (n D^{\mu-3})\text{ if } 2<\mu<3,\\
\Omega(\frac{M D^{3-1}}{\log D})=\Omega( \frac{n}{\log D}) \text{ if } \mu=3, \\
\end{cases} \] 
where we used the definition of $R$ in Eq.~\eqref{eq:R} and the fact that $M=\Theta(n/D^2)$.
Dividing by the optimal time $\Theta(n/D)$, we get 
\[\Over^{X^\mu}(n,D)=\begin{cases}
 \Omega(D^{\eps}) \text{ if } \mu=2+\epsilon, \mbox{~where~} 0<\epsilon<1, \\
 \Omega(\frac{D}{\log D})  \text{ if } \mu=3,
\end{cases}\]
as desired. 
This completes the proof of Theorem \ref{ThmLevyLowerLargeMu} in the continuous detection model.
\qed


\section{Scale-sensitivity of the intermittent Cauchy Walk}\label{sec:mainupper}

We take $n>4$ for technical reasons, and let $\lmax=\sqrt{n}/2$. 
As stated in the previous section, the overrun of the intermittent Cauchy walk $X^{\cauchy}$ for a target of diameter $D$ on the torus is $\Omega(\log n)$. The goal of this section is to prove the following theorem which states that this lower bound is nearly matched. 
\begin{theorem}\label{thm:upper-cauchy}\label{theorem:large-targets}
 Consider the Cauchy walk process $X^{\cauchy}$ on the torus $\T_n$. The hitting time of $X^{\cauchy}$ with respect to a target $S$ of diameter $D$ is 
\begin{linenomath*}\[t_{detect}^{X^\cauchy}(S)=O\left( \frac{n\log^3 n}{D} \right).\] \end{linenomath*}
Consequently, the overrun of $X^{\cauchy}$ for a target of diameter $D$ is $O(\log^3 D)$.
\end{theorem}
Theorem \ref{thm:upper-cauchy} concerns the Cauchy walk on the two-dimensional torus. As the  one-dimensional Cauchy walk is fairly well understood, 
it is tempting to analyze the two-dimensional walk by projecting it on the two axes and using the  
properties of the one-dimensional walk on these projections. However, this approach needs to somehow handle the fact that these projections are not independent of each other. As we could not find an easy way to overcome this dependence issue, we prove Theorem \ref{thm:upper-cauchy} following a different line of arguments, that directly examine the two-dimensional process.

To prove Theorem \ref{thm:upper-cauchy}, we can assume without loss of generality that the process starts at the origin, i.e., that $X^\cauchy(0)=0$.

Claim \ref{claim:timeVSmoves} implies
that in order to find the detecting time $t_{detect}^{X^{\cauchy}}(S)$ of $S$, it is sufficient to identify the expected number of steps until detecting $S$, as \begin{linenomath*}\[t_{detect}^{X^{\cauchy}}(S)=\E(m_{detect}^{X^{\cauchy}}(S))\cdot \tau=\Theta(\E(m_{detect}^{X^{\cauchy}}(S))\cdot \log n).\] \end{linenomath*}
Now let $Z$ be the process on $\R^2$ that evolves with the same steps $V(s)$ as $X^{\cauchy}$, i.e. $Z(m)=\sum_{s=1}^m V(s)$. Note that the projection of $Z$ on the torus 
$[-\sqrt{n}/2,\sqrt{n}/2]^2\subset \R^2$
is $X^{\cauchy}$.

The next lemma establishes a connection between $\E(m_{detect}^{X^{\cauchy}}(S))$ and the process $Z$ on $\R^2$. Given a set $S$, recall that $B(S)$ is the set of points at distance at most $1$ from $S$, and that $Z(m)$ detects $S$ if and only if $Z(m)\in B(S)$.

\rLemma{LemHitFromPointwise}{Consider a random walk process $Z$ on $\R^2$ and its projection $X$ on the torus $\T_n$ and denote by $Z^{z_0}$ the process $Z$ starting at $Z(0)=z_0$. Let $S\subset \T_n$. For any $m_0$,
	\begin{linenomath*}\begin{align}\label{eq:trick-applied}
	\E(m_{detect}^X(S))=O\left(m_0\cdot \frac{\sup_{z_0\in B(S)}\sum_{m=0}^{m_0} \Pr(Z^{z_0}(m)\in B(S) }{\sum_{m=m_0}^{2m_0}  \Pr(Z(m)\in B(S))}\right).
	\end{align}\end{linenomath*}}
We provide a formal proof of Lemma \ref{LemHitFromPointwise} in Section~\ref{app:lemma}. The proof is based on the technique in Adler et al.~\cite{hunter}, 
relying on the identity $\Pr(N\geq 1 ) = \frac{\E ( N)}{\E(N\mid N \geq 1)}$, that holds for any  non-negative random variable $N$.

Lemma \ref{LemHitFromPointwise} allows to deduce Theorem \ref{thm:upper-cauchy} from pointwise bounds on the Cauchy process $Z$ on $\R^2$, defined by Eq.~\eqref{eq:cauchy-precise-distribution-appendix}. The next lemma provides a lower bound on the p.d.f~$p^{Z(m)}$, of the process at step $m$.
\rLemma{LemLB}{For any constant $\alpha>0$, there exists a constant $c>0$ such that for any integer $m\in [1,\alpha \lmax]$, and any $x\in \R^2$, with $\norm{ x}\leq m$, \begin{linenomath*}\[p^{Z(m)}(x) \geq \frac{c}{m^2}.\] \end{linenomath*}}

From Lemma \ref{LemLB}, we immediately deduce  that the probability that $Z(m)$ detects a point $x\in \R^2$ is $\Omega(\int_{y\in B(x)} cm^{-2}dy)=\Omega(cm^{-2})$, where $B(x)=B(\{x\})$. This lower bound is complemented by the following upper bound.

\rLemma{LemUB}{
For any constant $\alpha>0$, there exists a constant $c'>0$ such that, for any integer $m\in [2,\alpha \lmax]$ and any $x\in \R^2$, we have
\begin{linenomath*}\[\Pr(Z(m)\in B(x))\leq \frac{c'\log^2 m}{m^2}.
\] \end{linenomath*}
}

Lemmas \ref{LemLB} and \ref{LemUB}
are formally proved in Section \ref{app:lemmas}. Let us give here a sketch of the proofs. 
Using the monotonicity property, the lower bound stated in Lemma \ref{LemLB}
follows once we prove that with at least some constant probability, the process at step $m$
belongs to the ring $\{x\mid \norm{x}\in [m,cm]\}$
for some constant $c>1$. This is because the area of this ring is roughly $m^2$, and each point in it is further from 0 than $x$, and hence, by monotonicity, less likely to be visited at step $m$. In order to establish the lower bound on the probability to be in the ring at step $m$, we first prove that with some constant probability, at some step before $m$, the walk goes to a distance at least $2m$.

Next, conditioning on that event, we prove that with a constant probability, the walk does not get much further away, i.e., it stays at a distance of at least $m$. To prove the latter claim, we use Chebyshev's inequality. It implies, for a one-dimensional process, that the distance traveled in $m$ steps is governed by $\sqrt{m}$ times the standard deviation of the step-length process. Here the standard deviation is too large (roughly $\sqrt{n}$), however, we can reduce it by conditioning on the event that none of the $m$ step-lengths are significantly larger than $m$, which occurs with a constant probability. Finally, we prove that by taking a sufficiently large constant $c$, it can be guaranteed that with a large constant probability,  the walk at step $m$ is at most at distance $cm$. Making sure that all these constant probability events happen simultaneously, we then establish the desired constant lower bound on the probability to be in the aforementioned ring at step $m$. 

For the proof of the upper bound in Lemma \ref{LemUB}, we first show that  because of the monotonicity property, it is sufficient to prove that the probability to detect $0$ at step $m$ is small, i.e., that 
\begin{linenomath*}\[\Pr(Z(m)\in B(0))= O\left(\frac{\log^2 m}{m^2}\right).\] \end{linenomath*}
Intuitively, to establish this, we first argue that with high probability in $m$, at some step before step $m$, the process has gone to a distance $d=\Omega(\frac{m}{\log m})$. By Corollary~\ref{cor:monotonicity}, the probability density function at any point in $B(0)$ would then be at most $O(\frac{1}{d^2})$, which is the desired bound.

\begin{proof}[Proof of Theorem \ref{theorem:large-targets}, assuming the aforementioned Lemmas]
Given the connected set $S$ of diameter $D\geq 1$, we first construct a subset $S'$, containing $\Theta(D)$ isolated points of $S$ that stretch over distance of roughly $D$, as follows. Take two points $u=(u_1,u_2)$ and $v=(v_1,v_2)$ in $S$ that are at distance $D$ from each other, so that $\max\{\lvert u_1-v_1\rvert,\lvert u_2-v_2\rvert\}\geq D/2$. Let us assume, without loss of generality, that $v_1-u_1\geq D/2$. 
Since $S$ is connected, for every $z\in [u_1,v_1]$, there exists $\phi(z)$ such that $(z,\phi(z))\in S$. 
Let $d=\lceil v_1-u_1 \rceil =\Theta(D)$. 
For integer $i\in\{0,1,\ldots,\lfloor d\rfloor\}$, define \begin{linenomath*}\[s(i)=(u_1+i,\phi(u_1+i)),\] \end{linenomath*} and let $S'=\{s(i)\mid  i=0,1,\ldots,\lfloor d\rfloor\}$. Note that $|S'|=\Theta(D)$.
Since $S'\subseteq S$, an upper bound on the detecting time of $S'$ is an upper bound on the detecting time of $S$. It is therefore sufficient to restrict attention to $S'$ and upper bound its detecting time. For that purpose we need to bound the time until visiting a point in $B(S')$, the set of points of distance at most $1$ from $S'$. Note that the area of $B(S')$ is $\lvert B(S')\rvert=\Omega(D)$. 
We also remark, that although $B(S')$ may not be connected, it may help the reader to imagine  $B(S')$ as a horizontal cylinder of length $\Theta(D)$ and radius $1$, i.e., to consider that $\phi(u_1+i)$ does not depend on $i$. Indeed, we will not require any condition on the $y$-coordinates of the $s(i)$'s.

In order to upper bound $\E(m_{detect}^{X^{\cauchy}}(B(S')))$
we shall apply Lemma \ref{LemHitFromPointwise} with $m_0=\sqrt{n}$. Note that $2m_0\leq \alpha \lmax$ for $\alpha=4$. We shall furthermore lower  bound  the denominator in the r.h.s of Eq.~\eqref{eq:trick-applied} and upper  bound  the numerator. Both these terms concern the Cauchy process $Z$ with cut off $\lmax$ on $\R^2$. 

Let us begin with the lower bound. With this setting of $m_0$, any $x\in B(S') \subseteq B(\T_n)\subseteq [-\sqrt{n}/2-1,\sqrt{n}/2+1]^2$  
trivially satisfies $\norm{ x} \leq m$, for any $m\geq m_0+1$, and we can apply 
 Lemma~\ref{LemLB} to get a lower bound on the denominator in the r.h.s of Eq.~\eqref{eq:trick-applied}:
\begin{linenomath*}\[ \sum_{m=m_0+1}^{2m_0}\Pr(Z(m)\in B(S') ) = \sum_{m=m_0+1}^{2m_0} \int_{x\in B(S')}p^Z_m(x)dx \geq \sum_{m=m_0+1}^{2m_0} \frac{c}{m^2}\lvert B(S')\rvert =\Omega\left( \frac{D}{\sqrt{n}}\right). \] \end{linenomath*}
Next, we provide an upper bound to the numerator of the r.h.s of Eq.~\eqref{eq:trick-applied} which is the number of returns to $S'$ conditioning on the fact that $Z(0)=z$, for some $z\in B(S')$. Let us denote this process by $Z^{z}$ (note that $Z=Z^0$). Then, 
\begin{linenomath*}\begin{align}
\sum_{m=0}^{m_0}\Pr(Z^{z}(m)\in B(S')) &\leq 2+\sum_{m=2}^{m_0} \Pr(Z^{z}(m)\in B(S')). \label{eq:returns-except-3}\end{align}\end{linenomath*}
Clearly, the probability density function $p^{Z^{z}(m)}$ of $Z^{z}(m)$ is obtained by a translation from $p^{Z(m)}$. Thus, by Corollary~\ref{cor:monotonicity}, we have for any $y\in \R^2$:
\begin{linenomath*}\[p^{Z^{z}(m)}(y) \leq \frac{1}{\norm{ y-z}^2}.
\] \end{linenomath*}
In particular, for 
$y$ such that $\norm{y-z}\geq 2$,
\begin{linenomath*}\begin{equation}\label{eq:yfar}
\Pr(Z^z(m)\in B(y))\leq \frac{1}{(\norm{ y-z}-1)^2},\end{equation}\end{linenomath*}
since every $w\in B(y)$ satisfies $\norm{w-z}\geq \norm{y-z}-1\geq 0$.

Next, as $z\in B(S')$, consider an index $i_z\in \{0,\dots,d-1\}$ for which $z\in B(s(i_z))$. Let $r_m=\frac{m}{\sqrt{c}\log m}$ with $c$ being the constant $c'$ mentioned in Lemma \ref{LemUB}. To exploit Eq.~\eqref{eq:yfar}, we define \begin{linenomath*}\[I=\{i\in \{0,\ldots,d-1\} \mid \lvert s(i)_1-s(i_z)_1\rvert= \lvert i-i_z\rvert \leq r_m+2\},\] \end{linenomath*}
 and $I^c=\{0,\ldots,d-1\}\setminus I$. We proceed with the following decomposition:
\begin{linenomath*}\begin{align}
\Pr(Z^{z}(m)\in B(S'))\leq \sum_{i\in I} \Pr\left(Z^{z}(m)\in B(s\small(i\small))\right) + \sum_{i\in I^c}\Pr\left(Z^{z}(m)\in B(s\small(i\small))\right). \label{eq:returns-decompose-according-to-ub}
\end{align}\end{linenomath*}
By construction, $\lvert I\rvert \leq 2(r_m+2)+1$. 
Hence, using Lemma~\ref{LemUB}, the first sum in the r.h.s of Eq.~\eqref{eq:returns-decompose-according-to-ub}
is at most:
\begin{linenomath*}\begin{align*}
\sum_{i\in I} \Pr(Z^{z}(m)\in B(s(i)))  \leq \frac{\lvert I\rvert }{r_m^2}=O\left(\frac{1}{r_m}\right).
\end{align*}\end{linenomath*}
Next, we aim to upper bound the sum on $I^c$.
By the triangle inequality, for any $i\in I^c$, we have $\norm{ s(i)-z}\geq
\norm{ s(i)-s(i_z)}-1\geq \lvert i-i_z\rvert -1 >1$. Hence, by Eq.~\eqref{eq:yfar},
 we get:
\begin{linenomath*}\begin{align*}
    \sum_{i\in I^c}\Pr(Z^{z}(m)\in B(s(i)))&\leq
    \sum_{i\in I^c}
 \frac{1}{(\norm{ s(i)-z}-1)^2}
 \\& \leq \sum_{i\in I^c}
 \frac{1}{(\lvert i-i_z\rvert-2)^2}
 \\& \leq
 \sum_{k\in \Z, \lvert k\rvert \geq \lceil r_m\rceil }\frac{1}{k^2}
    = O\left(\frac{1}{r_m}\right),
    \label{eq:returns-far-points}
\end{align*}\end{linenomath*}
where we used in the last line that $i\in I^c\subset \{i_z+k\mid k\in \Z \mbox{~and~} \lvert k\rvert> r_m+2 \}$. Thus, we get by Eq.~\eqref{eq:returns-decompose-according-to-ub}:
\begin{linenomath*}\[
\Pr(Z^{z}(m)\in B(S'))= O\left(\frac{1}{r_m}\right).
\] \end{linenomath*}
Plugging this in Eq.~\eqref{eq:returns-except-3}, together with the definition $r_m=\frac{m}{\sqrt{c}\log m}$, and the fact that $m_0=O(\sqrt{n})$, we get:
\begin{linenomath*}\begin{align*}
    \sum_{m=0}^{m_0}\Pr(Z^z(m)\in B(S')) =2+O\left(\sum_{m=2}^{m_0}\frac{\log m}{m}\right)=O(\log^2 n),
\end{align*}\end{linenomath*}
which stands for any $z\in B(S')$. Altogether, the fraction in Eq.~\eqref{eq:trick-applied} satisfies:
\begin{linenomath*}\begin{align*}
	\frac{\sup_{z\in B(S')}\sum_{m=0}^{m_0} \Pr(Z^z(m)\in B(S'))}{\sum_{m=m_0}^{2m_0}  \Pr(Z(m)\in B(S'))} = O \left ( \frac{\sqrt{n}}{D}\cdot {\log^2  n}\right ).
\end{align*}\end{linenomath*} 
Together with the fact that $m_0=O(\sqrt{n})$,  Lemma \ref{LemHitFromPointwise} implies that $\E(m_{detect}^X(S))= O(\frac{n}{D}\log^2 n)$. Finally, using Claim \ref{claim:timeVSmoves} and the fact that $\tau=\Theta(\log n)$, we have
\begin{linenomath*}\[t_{detect}^{X}(S)=O\left (\frac{n\log^3 n}{D}\right ), \] \end{linenomath*}
and since this is true for any connected set $S\subseteq \T_n$ of diameter $D$, we obtain $t_{detect}^{X}(n,D)=O\left (\frac{n\log^3 n}{D}\right )$, as desired.
\end{proof}

\subsection{Proof of Lemma \ref{LemHitFromPointwise}}\label{app:lemma}

The goal of this section is to prove of Lemma \ref{LemHitFromPointwise}. Recall, we
consider a random walk process $Z$ on $\R^2$ and its projection $X$ on the torus $\T_n$. Let $S\subset \T_n$. Our goal is to show that for any $m_0$,
	\begin{linenomath*}\begin{align}
	\E(m_{detect}^X(n,D))=O\left(m_0\frac{\sup_{z_0\in B(S)}\sum_{m=0}^{m_0} \Pr(Z^{z_0}(m)\in B(S))}{\sum_{m=m_0}^{2m_0}  \Pr(Z(m)\in B(S))}\right). \end{align}
	\end{linenomath*}
\begin{proof}

We begin with the following claim that shows that if the probability to detect $S$ by step $m$ is at least $p$ for any starting point, then the expected detecting step is at most $m/p$. 
The claim will then be used to prove the lemma by showing that the inverse of the supremum in Eq.~\eqref{eq:trick-applied} is a lower bound for $\Pr(m^X_{detect}(S)\leq 2m_0)$. 
\begin{claim}	\label{lem:hit-time-from-point-probability}Fix an integer $m>0$ and a real number $q>0$ and a set $S\subseteq \T_n$. Denote by $X^x$ the process $X$ starting at $X(0)=x$. If, for any $x\in \T_n$, we have $\Pr(m^{X^x}_{detect}(S)\leq m) \geq q$ then $\E(m^X_{detect}(S))\leq mq^{-1}$.
\end{claim}
\begin{proof}[Proof of Claim \ref{lem:hit-time-from-point-probability}]The proof of the claim is simple. Given a set $S$,
define a Bernoulli variable $\chi$ as follows. 
Consider $m$ steps of the process and define $\chi$ to be ``success'' if and only if the process hits $S$ within these $m$ steps. Note that $\chi$ has probability at least $q$ to be ``success'' regardless of where the process starts, by hypothesis.  Hence, the expected number of trials until $\chi$ succeeds is at most $1/q$. This translates to  $\E(m^X_{detect}(S))\leq mq^{-1}$, and establishes Claim \ref{lem:hit-time-from-point-probability}.
\end{proof}

To conclude the proof of Lemma \ref{LemHitFromPointwise}, relying on Claim \ref{lem:hit-time-from-point-probability}, it is sufficient to prove that, for any $S\subset \T_n$,
\begin{linenomath*}\begin{equation}\label{eq:proba-X-hits-x}
\Pr(m^X_{detect}(S)\leq 2m_0) \geq \frac{\sum_{m=m_0}^{2m_0}  \Pr(Z(m)\in B(S))}{\sup_{z_0\in B(S)}\sum_{m=0}^{m_0} \Pr(Z^{z_0}(m)\in B(S) )} .\end{equation}\end{linenomath*}

 For this, we rely on the following identity (see also Adler et al.~\cite{hunter}). 
 If $N$ is a non-negative random variable then:
\begin{linenomath*}\begin{align}
\Pr(N\geq 1 ) = \frac{\E ( N)}{\E(N\mid N \geq 1)}.
\label{eq:trick}
\end{align}\end{linenomath*} 
We employ this identity for the random variable $N_S(m_0,2m_0)$ which is the number of times $Z$ visits $B(S)$ between steps $m_0$ and $2m_0$ included. Note that this quantity is positive if and only if $B(S)$ is visited during this interval by $Z$. Moreover, since $S\subset \T_n$ and $X$ is the projection of $Z$ on the torus, then  $Z(m)\in B(S)$ implies that also $X(m)\in B(S)$. 
Therefore,
\begin{linenomath*}\begin{equation} \label{eq:probability-hit-x-N}\Pr ( m^X_{detect}(S) \leq 2m_0) \geq \Pr \left (N_S(m_0,2m_0) \geq 1 \right ).\end{equation}\end{linenomath*}
Note that $N_S(m_0,2m_0)=\sum_{m=m_0}^{2m_0}\mathbf{1}_{Z(m)\in B(S)} $, so that \begin{linenomath*}\begin{equation}\label{eq:E_x}\E ( N_S(m_0,2m_0) )=\sum_{m=m_0}^{2m_0} \Pr( Z(m)\in B(S)).\end{equation}\end{linenomath*} Note also that the denominator in Eq.~\eqref{eq:trick} applied to $N_S(m_0,2m_0)$ verifies
 \begin{linenomath*}\begin{align*} \E \left ( N_S(m_0,2m_0) \mid N_S(m_0,2m_0) \geq 1\right )&= \E \left ( N_S(m_0,2m_0) \mid  Z(m)\in B(S) \text{ for some } m\in [m_0,2m_0]\right ) \\
 & \leq \sup_{z_0\in B(S)}\E \left ( N_S(m_0,2m_0) \mid  Z(m_0)=z_0 \right ) 
 \\ &\leq \sup_{z_0\in B(S)}\E \left ( N_S(0,m_0) \mid  Z(0)=z_0 \right ), \end{align*}\end{linenomath*}
 where the first inequality comes from the fact that visiting $B(S)$ earlier (i.e., for $m=m_0$ instead of $m>m_0$) can only increase the number of returns to $B(S)$, and the second inequality is a consequence of the Markov property. Finally, write, as above,
 \begin{linenomath*}\begin{equation}\label{eq:E_0}
     \sup_{z_0\in B(S)}\E \left ( N_S(0,m_0) \mid  Z(0)=z_0 \right ) = \sup_{z_0 \in B(S)}\sum_{m=0}^{m_0} \Pr( Z^{z_0}(m)\in S). 
 \end{equation}\end{linenomath*}
Therefore, when applied to $N_S(m_0,2m_0)$, Eq.~\eqref{eq:trick}, combined with Eqs. \eqref{eq:probability-hit-x-N}, \eqref{eq:E_x} and \eqref{eq:E_0}, implies that
\begin{linenomath*}\begin{equation*}
{\Pr}_0(m^X_{detect}(S) \leq 2m_0) \geq \frac{\sum_{m=m_0}^{2m_0}  \Pr( Z(m)\in B(S))}{\sup_{z_0\in B(S)}\sum_{m=0}^{m_0} {\Pr}( Z^{z_0}(m)\in B(S))}.
\end{equation*} \end{linenomath*}
This establishes Eq.~\eqref{eq:proba-X-hits-x}, and thus completes the proof of Lemma~\ref{LemHitFromPointwise}. \end{proof}

\subsection{Proofs of Lemmas \ref{LemLB} and \ref{LemUB}}\label{app:lemmas}
In this section we aim to prove the following lower and upper bounds, stated in Lemmas \ref{LemLB} and \ref{LemUB}, respectively. The proof of Lemma \ref{LemLB} is given in 
Section \ref{sec:lowerbound}, and the proof of Lemma \ref{LemUB} is given in Section \ref{sec:upperbound}. Before presenting these proofs, let us first 
 first establish lower and upper bounds on the distance traveled by the walk at step $m$.
\subsubsection{Superdiffusive properties of the Cauchy walk on \texorpdfstring{$\R^2$}{the plane}}
We first remark that the probability to choose a length in a given interval is easily computed from Eq.~\eqref{eq:cauchy-precise-distribution-appendix}. 
\begin{observation}\label{obs:tail-law}
 The probability to do a step of length at most $\ell> 0$ is $a\ell$ if $\ell \leq 1$ and $a(2-\frac{1}{\ell})$ if $\ell >1$. For integers $\lmax\geq\ell_2 \geq \ell_1\geq 1$, the probability to choose a length in $[\ell_1,\ell_2]$ is $a( \frac{1}{\ell_1}-\frac{1}{\ell_2})$.
 \end{observation}
 
The next claim quantifies the probability that the Cauchy process goes to a distance of at least $d$ after $m$ steps.
In particular, it shows that in step $m$, the process is at a distance of $\Omega(m)$ with constant probability, and that it is at a distance of $\Omega(m/\log m)$ with high probability in $m$.

\begin{claim}\label{claim:X-far}For any integer $m\geq 2$ and any real $d\in [1,\frac{\lmax}{3}]$  we have, \begin{linenomath*}\[\Pr\left(\exists s\leq m \mbox{~s.t.~} \norm{ Z(s) }\geq d
\right)\geq 1-e^{-cm/d},\] \end{linenomath*}
for some constant $c>0$.
 In particular this lower bound is at least 
 \begin{itemize}
    \item $1-O(m^{-2})$ if $d=c'\frac{m}{\log m}$ with $c'>0$ a small enough constant,
    \item $\Omega(1)$ if $d=c'm$ for any constant $c'>0$ with $c'm\leq \lmax/3$.
\end{itemize}
\end{claim}
\begin{proof}
By Observation \ref{obs:tail-law}, the probability that a given step has a length at least $2d$ is $a(\frac{1}{2d}-\frac{1}{\lmax})\geq \frac{a}{6d}$.
Since the steps are independent, the probability of the event $\mathcal{A}$ that at least one of the steps $1,\dots,m$ has a length at least $2d$ is
\begin{linenomath*}\[\Pr(\mathcal{A})\geq 1-\left(1-\frac{a}{6d}\right)^{m}. \] \end{linenomath*}
Writing $\left(1-a/6d\right)^{m}=e^{m\log(1-\frac{a}{6d})}\leq e^{-cm/d}$, for some constant $c>0$, we get 
 \begin{linenomath*}\[\Pr(\mathcal{A})\geq 1-e^{-cm/d}. \] \end{linenomath*}
 To conclude, it suffices to show that  $\mathcal{A}$ implies that there exists a step $s\leq m$ for which $\norm{ Z(s)}\geq d$. 
Indeed, suppose that $\mathcal{A}$ occurs and let $s\leq m$ be the first step of length $2d$ or more.
Then,
\begin{itemize}
    \item Either $\norm{ Z(s-1)}\geq d$, in which case we are done.
    \item Or $\norm{ Z(s-1)} < d$. In this case, as $Z(s)=Z(s-1)+V(s)$, we have $\norm{ Z(s)} \geq \norm{ V(s)} - \norm{ Z(s-1)} > 2d-d=d$.
\end{itemize}
This concludes the proof of Claim \ref{claim:X-far}.
\end{proof}

Claim \ref{claim:X-far} asserts that, with some probability, the walk goes far from $0$. Conversely, the next claim says that with some constant probability, the walk does not get too far.

\begin{claim}\label{claim:X-m-less-than-cm}
\begin{itemize}
    \item For any constant $c>0$, there exists a constant $\delta>0$ such that, for any two integers $1\leq s\leq m$, we have $\Pr(\norm{ Z(s)} \leq cm )\geq \delta$.
    \item For any constant $0<\delta<1$, there exists a (large enough) constant $c>0$ such that, for any two integers $1\leq s\leq m$, we have $\Pr(\norm{ Z(s)} \leq cm )\geq \delta$.
\end{itemize}
\end{claim}
\begin{proof}Fix an integer $m\geq 1$ and let $c''$ be a constant, to be chosen later. Let $\mathcal{A}$ denote the event that each of the first $m$ steps has length at most $\ell=c'' m$. We have, for any integer $s\leq m$, and any constant $c>0$,
\begin{linenomath*}\begin{equation}
    \label{eq:norm-X-cond-small-steps}\Pr(\norm{ Z(s)} \leq cm ) \geq \Pr(\mathcal{A}) \cdot\Pr(\norm{ Z(s)} \leq cm  \mid\mathcal{A}).
\end{equation}\end{linenomath*}  
We shall study separately each term in the r.h.s~of Eq.~\eqref{eq:norm-X-cond-small-steps}, and establish the following:
\begin{itemize}
    \item For the first item of Claim \ref{claim:X-m-less-than-cm}, we shall take $c''>0$ so that both factors are constants (hence their multiplication is at least some constant $\delta)$,
    \item For the second item of Claim \ref{claim:X-m-less-than-cm}, where the bound $\delta$ is given, we will show that both terms can be made at least $\sqrt{\delta}$ by choosing $c$ and $c''$ appropriately.
\end{itemize} 
 Proceeding with the first term in the r.h.s~of Eq.~\eqref{eq:norm-X-cond-small-steps}, by Observation \ref{obs:tail-law}, we have:
\begin{linenomath*}\[\Pr(\mathcal{A})=\begin{cases}(ac''m)^m \text{ if } c''m\leq 1 \\ 
(2a)^m(1-\frac{1}{2c''m})^m \text{ if } c''m\in [1,\lmax]\\
1 \text{ if } c''m\geq \lmax \end{cases}.\] \end{linenomath*}
For $1\leq m\leq \frac{1}{c''}$, we have $(ac''m)^m \geq (ac''m)^{\frac{1}{c''}}$ as $ac''m\leq c''m \leq 1$, and $ (ac''m)^{\frac{1}{c''}}\geq (ac'')^{\frac{1}{c''}}$ as $m\geq 1$. For the second item,  note that the function $(1-\frac{\alpha}{x})^x=e^{x\log(1-\frac{\alpha}{x})}$ is increasing in $x\geq \alpha$ and thus, for $x\geq 2\alpha$, we have $(1-\frac{\alpha}{x})^x\geq 2^{-2\alpha}$. Applying this with $\alpha=\frac{1}{2c''}$, we have, $(1-\frac{1}{2c''m})^m\geq 2^{-\frac{1}{c''}}$, for $m\geq \frac{1}{c''}$. Overall, using $2a\geq 1$, we get 
\begin{linenomath*}\[\Pr(\mathcal{A})\geq \begin{cases}(\frac{c''}{2})^{\frac{1}{c''}} \text{ if } c''m\leq 1 \\ 
2^{-\frac{1}{c''}} \text{ if } c''m\in [1,\lmax]\\
1 \text{ if } c''m\geq \lmax\end{cases}.\] \end{linenomath*} 
Hence,
\begin{itemize}
    \item $\Pr(\mathcal{A})=\Omega(1)$ for any given $c''>0$.
    \item Furthermore, with respect to the second item of Claim \ref{claim:X-m-less-than-cm} where $0<\delta<1$ is given, we can choose $c''$ large enough (in particular, we take $c''\geq 1$ so that $c''m\geq 1$), to ensure that $\Pr(\mathcal{A})\geq 2^{-\frac{1}{c''}} \geq \sqrt{\delta}$.
\end{itemize} 
We are now ready to lower bound the second factor in  Eq.~\eqref{eq:norm-X-cond-small-steps}, namely, $\Pr(\norm{ Z(s)} \leq cm  \mid\mathcal{A})$. We begin with a notation: If $X$ is a random variable, let us write $X^{\mathcal{A}}$ for the random variable $X$ conditioned on the occurrence of $\mathcal{A}$. Our first goal is to prove that 
\begin{linenomath*}\begin{equation}\label{eq:Z-Chebyshev}
    \Pr(\norm{ Z^{\mathcal{A}}(s)} \leq cm  ) \geq 1-\frac{8s\E(\norm{V^{\mathcal{B}}}^2)}{c^2m^2},
\end{equation}\end{linenomath*}
where $V^{\mathcal{B}}=(V_1^{\mathcal{B}},V_2^{\mathcal{B}})$ is one step-vector of the walk on $\R^2$, conditioned on the event $\mathcal{B}$ that it is at most $c''m$.
Eq.~\eqref{eq:Z-Chebyshev} will be established by applying Chebyshev's inequality on each of the projections on the axes and using a union bound argument. Specifically, decomposing the walk $Z$ on the two axes, by writing $Z=(Z_1,Z_2)$, we first use a union bound to obtain:
\begin{linenomath*}\begin{align*} \Pr(\norm{ Z^{\mathcal{A}}(s)} > cm  ) &\leq \Pr(\exists i=1,2 \mbox{~s.t.~} \lvert Z_i^{\mathcal{A}}(s)\rvert  > cm/2)\\
&\leq \Pr(\lvert Z_1^{\mathcal{A}}(s)\rvert  > cm/2)+\Pr(\lvert Z_2^{\mathcal{A}}(s)\rvert  > cm/2)\\
&\leq 2\Pr( \lvert Z_1^{\mathcal{A}}(s)\rvert  > cm/2),
\end{align*}\end{linenomath*}
where we used the symmetry to deduce that $Z_1$ and $Z_2$ share the same distribution.
Hence, 
\begin{linenomath*}\[ \Pr(\norm{ Z^{\mathcal{A}}(s)} \leq cm ) \geq 1-2\Pr(\lvert Z_1^{\mathcal{A}}(s)\rvert  > cm/2). \] \end{linenomath*}
Next, we aim to lower bound the r.h.s.
 Relying on the fact that the expectation of $Z_1^{\mathcal{A}}(s)$ is 0 for any $s$, 
by Chebyshev's inequality, we have: 
\begin{linenomath*}\[ \Pr(\lvert Z_1^{\mathcal{A}}(s) \rvert > cm/2 )\leq \frac{4\Var(Z_1^{\mathcal{A}}(s) )}{c^2m^2}.
\] \end{linenomath*}
Since $Z_1^{\mathcal{A}}(s)$ is the sum of $s$ independent steps that follow the same law as $V_1^{\mathcal{B}}$, we have: 
\begin{linenomath*}\[ \Var(Z_1^{\mathcal{A}}(s))= s\Var(V_1^{\mathcal{B}}).\] \end{linenomath*}
As the expectation of $V_1^{\mathcal{B}}$ is zero, we have $\Var(V_1^{\mathcal{B}})=\E((V_1^{\mathcal{B}})^2)$. Furthermore, since $\lvert V_1^{\mathcal{B}}\rvert \leq \norm{V^{\mathcal{B}}}$, we obtain:
\begin{linenomath*}\[ \Var(Z_1^{\mathcal{A}}(s))\leq s\E(\norm{V^{\mathcal{B}}}^2),\] \end{linenomath*}
which concludes the proof of Eq.~\eqref{eq:Z-Chebyshev}.
Next, let us estimate $\E(\norm{V^{\mathcal{B}}}^2)$. If, on the one hand, $c''m\leq 1$, then, when conditioning on $\mathcal{A}$, the length of a step is chosen uniformly at random in $[0,c''m]$. Thus, its second moment is \begin{linenomath*}\begin{equation}\label{eq:norm-V-m-small}
    \E(\norm{V^{\mathcal{B}}}^2)=\int_{0}^{c''m}\ell^2 \frac{d\ell}{c''m}=\frac{(c''m)^2}{3}.
\end{equation}\end{linenomath*}
On the other~hand, if $c''m\geq 1$, then $V^{\mathcal{B}}$ is a Cauchy walk with cut off $\ell_{max}=c''m$. Hence, its second moment is
\begin{linenomath*}\begin{align}\nonumber
   \E(\norm{V^{\mathcal{B}}}^2)&= a'\int_0^1 \ell^2d\ell +a'\int_1^{c''m} \ell^2\ell^{-2}d\ell \\
     &\leq a'\int_0^{c''m} 1d\ell=a'c''m\leq c''m.\label{eq:norm-V-m-big}
\end{align}\end{linenomath*}
Overall, by Eqs.~\eqref{eq:Z-Chebyshev}, \eqref{eq:norm-V-m-small} and \eqref{eq:norm-V-m-big} we find that, for $s\leq m$, 
\begin{linenomath*}\begin{align*}
    \Pr(\norm{ Z^{\mathcal{A}}(s)} \leq cm  )\geq \begin{cases}1-\frac{8sc''^2}{3c^2} \text{ if } c''m\leq 1\\
    1-\frac{8sc''}{c^2m} \text{ if } c''m\geq 1
    \end{cases}\\
    \geq \begin{cases}1-\frac{8c''}{3c^2} \text{ if } c''m\leq 1\\
    1-\frac{8c''}{c^2} \text{ if } c''m\geq 1
    \end{cases}.
\end{align*}\end{linenomath*}
We then conclude the proof of Claim \ref{claim:X-m-less-than-cm} by observing the following.
\begin{itemize}
    \item For the first item of Claim \ref{claim:X-m-less-than-cm}, we have proved that $\Pr(\mathcal{A})=\Omega(1)$ for any constant $c''>0$. Hence, we may now choose $c''$ small enough so that $\Pr(\norm{ Z^{\mathcal{A}}(s)} \leq cm  )=\Omega(1)$.
    \item For the second item of Claim \ref{claim:X-m-less-than-cm}, we have already chosen $c''$ to be large (in order to have $\Pr(\mathcal{A})\geq \sqrt{\delta}$, but we are free to choose $c$ large enough so that $\Pr(\norm{ Z^{\mathcal{A}}(s)} \leq cm  )\geq \sqrt{\delta}$.
\end{itemize}
\end{proof}

\subsubsection{Proof of Lemma \ref{LemLB} (lower bound)}\label{sec:lowerbound}

In this section we prove the following:
\LemLB
\begin{proof}
First note that for $m=1$, the lemma holds by the definition of the L\'evy process. Let us therefore consider an integer $m\geq 2$.

By the monotonicity property (Corollary \ref{cor:monotonicity}), it is enough to prove that there is some constant $c'>1$ such that,
\begin{linenomath*}\begin{equation}\label{eq:far}
\Pr(m\leq \norm{Z(m)} \leq c'm) =\Omega(1).
\end{equation}\end{linenomath*}
Indeed, if this holds, then, since the area of the ring $\{ y\in \R^2 \text{ s.t. } m\leq \norm{y} \leq c'm \}$ is $\Theta(m^2)$, then we would have that for at least one point $u$ in this ring,  $p^{Z(m)}(u)=\Omega(m^{-2})$. Then, 
 by monotonicity, 
for $x\in \R^2$ such that $\norm{x}\leq m$, we would have $p^{Z(m)}(x)\geq p^{Z(m)}(u)=\Omega(m^{-2})$ which is the desired lower bound.

We thus proceed to prove Eq.~\eqref{eq:far}. For this, let us define,  for a given $m\in [2,\alpha\lmax]$, the event \begin{linenomath*}\[\mathcal{A}_{far} =\exists s\leq m \mbox{ s.t.} \norm{Z(s)}\geq 2m .\] \end{linenomath*}
We next prove the following claim.
\begin{claim}\label{claim:zm-more-than-2m}
 $\Pr(\mathcal{A}_{far})=\Omega(1)$, where the constant in lower bound does not depend on $m$.
\end{claim}
\begin{proof}[Proof of Claim \ref{claim:zm-more-than-2m}]
By Claim \ref{claim:X-far}, we immediately get that the claim holds for any $m\in [2, \lmax/6]$.
We next show that the claim holds also for $m\in [\lmax/6,\alpha \lmax]$.  Intuitively, we prove this using a constant number of iterations.
Each iteration consists of at most $m'= \alpha' \lmax$ steps, with $\alpha'$ a small constant, during which we are guaranteed to go a distance of $\lmax/3$ with constant probability. Because the direction is chosen uniformly at random, at the cost of reducing this probability by a constant factor, we can further impose that 
the $x$-coordinate increases by a factor of, say, $\lmax/5$. As these iterations are independent, and since $\alpha$ is a constant, we can guarantee that  up to step $m=\alpha \lmax$,  the process goes away to a distance of at least $2\alpha \lmax$ with constant probability.

Formally, first notice that we can take $\alpha>1$ without loss of generality. Note now that since  $m\in [\lmax/6,\alpha \lmax]$, the second item in Claim~\ref{claim:X-far} implies that: 
\begin{linenomath*}\[ \Pr\left(\exists s\leq \frac{m}{10\alpha}\text{ s.t. } \norm{Z(s)}\geq \frac{\lmax}{3} \right) \geq c'_\alpha ,  \] \end{linenomath*} 
for some constant $c'_\alpha>0$. As a consequence, since the direction of $Z(s)$ is distributed uniformly at random, we have:
\begin{linenomath*}\begin{equation}\label{eq:make-it-repeat}\Pr\left(\exists s\leq\frac{m}{10\alpha}, Z_1(s) \geq \frac{\lmax}{4}\right) \geq c_\alpha,
\end{equation}\end{linenomath*} for some constant $c_\alpha>0$. When this occurs, let $s_1\leq \frac{m}{10\alpha}$ be such that $Z_1(s_1) \geq \frac{\lmax}{4}$. By the Markov property, starting from step $s_1$, we can then apply again \eqref{eq:make-it-repeat} to show that with probability $c_\alpha$, there is a $s_2\leq s_1+\frac{m}{10\alpha}\leq 2\frac{m}{10\alpha}$
such that $Z_1(s_2)\geq Z_1(s_1)+\frac{\lmax}{4}\geq 2\frac{\lmax}{4}$. Overall, this happens with probability $c_\alpha^2$. Repeating this $\lceil 9\alpha\rceil$ times, we finally get:
\begin{linenomath*}\[ \Pr\left(\exists s\leq\lceil 9\alpha\rceil \frac{m}{10\alpha}, Z_1(s) \geq \lceil 9\alpha\rceil\frac{\lmax}{4}\right) \geq c_\alpha^{\lceil 9\alpha\rceil},\] \end{linenomath*}
which is a positive constant. Because $\alpha>1$, this implies $\Pr(\exists s\leq m, Z_1(s) \geq 2\alpha\lmax )=\Omega(1)$. As $2\alpha\lmax\geq 2m$ and $\norm{Z}(s)\geq \lvert Z_1(s)\rvert$, this, in turn, implies $\Pr(A_{far})=\Omega(1)$, completing the proof of Claim~\ref{claim:zm-more-than-2m}.\end{proof}

Next, conditioning on $\mathcal{A}_{far}$, we write:
\begin{linenomath*}\begin{align}\label{eq:firstm}
\Pr(\norm{Z(m)} \geq m \mid \mathcal{A}_{far})&\geq \min_{s\leq m} \Pr(\norm{Z(m)} \geq m \mid \norm{Z(s)}\geq 2m )\\
     &\geq \min_{s\leq m} \Pr(\norm{Z(m-s)} \leq  m), \label{eq:lastm}
\end{align}\end{linenomath*}
where we used the Markov property, and the spatial homogeneity of the process, in the latter inequality. In words, in the r.h.s.~of Inequality~\eqref{eq:firstm}, we examine the probability to be at a high distance (i.e., $m$), knowing that the process was even further (at some point $x$ at distance at least~$2m$). In Inequality~\eqref{eq:lastm} we bound this by the probability of staying within distance $m$.

 By the first item of Claim \ref{claim:X-m-less-than-cm},
the r.h.s~of Inequality~\eqref{eq:lastm} is at least some positive constant (again, independent of $m$). Overall, for any $m\geq 2$, we have:
\begin{linenomath*}\begin{align*}
     \Pr(\norm{Z(m)} \geq m ) \geq \Pr(\norm{Z(m)} \geq   m \mid \mathcal{A}_{far} ) \cdot 
     \Pr(\mathcal{A}_{far}) \geq \gamma,
\end{align*}\end{linenomath*}
for some constant $\gamma>0$ (independent of $m$). Next, using the second item of Claim \ref{claim:X-m-less-than-cm}, with $\delta=1-\frac{\gamma}{2}$, we get that there exists a large enough constant $c'>0$ (again, independent of $m$), such that: 
\begin{linenomath*}\begin{equation}\label{eq:alpha'}
    \Pr(\norm{ Z(m)} \leq c' m) \geq \delta.
    \end{equation}\end{linenomath*}
Hence, using a union bound argument, we have:
\begin{linenomath*}\begin{align*} 
\Pr(m\leq \norm{Z(m)} \leq c'm)
&\geq \Pr(\norm{Z(m)} \geq  m) + \Pr(\norm{ Z(m)} \leq c'm) -1 \\ 
&\geq \gamma+\delta-1=\frac{\gamma}{2}>0.
\end{align*}\end{linenomath*}
This establishes Eq.~\eqref{eq:far} and thus concludes the proof of Lemma \ref{LemLB}.
\end{proof}

\subsubsection{Proof of Lemma \ref{LemUB} (upper bound)}\label{sec:upperbound}

This section is dedicated to the proof of Lemma \ref{LemUB}:
\LemUB
\begin{proof}
Let $\alpha>0$ and $m\in [2,\alpha \lmax]$. Due to the 
 monotonicity property stated in Corollary \ref{cor:monotonicity}, it is sufficient to prove this result for $x=0$. Indeed, for any $x\in \R^2$,  the sets $B(0)\setminus B(x)$ and $B(x)\setminus B(0)$ have the same area $A$, and
 
\begin{linenomath*}\begin{align*} \Pr\left (Z(m)\in B(x)\setminus B(0)\right )&\leq A \max_{y\in B(x)\setminus B(0)}  \{ p^{\norm{Z(m)}}(y)\}   \\& \leq A\min_{y\in B(0)\setminus B(x)}\rvert \{ p^{\norm{Z(m)}}(y)\} \\& \leq \Pr\left (Z(m)\in B(0)\setminus B(x)\right ),
\end{align*}\end{linenomath*}
 where the second inequality is due to the monotonicity property and the fact that any point in  $B(x)\setminus B(0)$ is at distance more than 1 from the origin, and hence, further from 0 than any point in $B(0)\setminus B(x)$.
 This shows that $\Pr(Z(m)\in B(x))\leq \Pr(Z(m)\in B(0))$, hence it is 
 sufficient to prove the required upper bound for $x=0$.

Intuitively, to establish this, we say that with high probability, there is some step $s\leq m$ for which $Z(s)$ is ``distant'' (at least $cm/\log m$). Conditioning on this, the probability to be located in $B(0)$ at step $m$ is found out to be small, due to the monotonicity of the process (Corollary \ref{cor:monotonicity}). Formally, consider a (small) positive constant $c$, and let $\mathcal{A}$ be the event that there is some $s\leq m$ for which $\norm{Z(s)}\geq cm/\log m$.

Consider $B(0)$ the ball of radius $1$ with center $0$. Write
\begin{linenomath*}\begin{align} \Pr(Z(m)\in B(0))&= \Pr(Z(m)\in B(0) \cap \mathcal{A})+\Pr(Z(m)\in B(0) \cap \neg \mathcal{A}) \nonumber \\
&\leq \Pr(Z(m)\in B(0)\mid \mathcal{A})+\Pr(\neg \mathcal{A}), \label{eq:Z-in-M-cond-distant}
\end{align}\end{linenomath*}
By the first item of Claim \ref{claim:X-far}, taking $c$ to be sufficiently small, we have 
\begin{linenomath*}\[\Pr(\neg \mathcal{A})=O(m^{-2}).\] \end{linenomath*}
In order to express the remaining term of Eq.~\eqref{eq:Z-in-M-cond-distant}, we will denote in the following equation $Z^x$ the Cauchy process on $\R^2$ with cut off $\ell_{max}$ starting with $Z(0)=x$. Since our process was defined to start at $0$, we have $Z=Z^0$. Remark that the law of $Z^x$ is obtained by a translation of that of $Z^0$. With this notation in mind, we have, using the Markov property for the second inequality:
\begin{linenomath*}\begin{align*}
  \Pr(Z^0(m)\in B(0)\mid \mathcal{A})  &\leq \max_{s\leq m}\Pr(Z^0(m)\in B(0)\mid \norm{Z^0(s)}\geq cm/\log m) \\
  &\leq \max_{s\leq m}\sup_{\norm{x}\geq cm/\log m} \Pr(Z^x(m-s)\in B(0))  \\
 &= \max_{s\leq m}\sup_{\norm{x}\geq cm/\log m} \Pr(Z^x(s)\in B(0)) \\
 &= \max_{s\leq m}\sup_{\norm{x}\geq cm/\log m} \Pr(Z^0(s)\in B(-x) )\\
 &= \max_{s_\leq m}\sup_{\norm{x}\geq cm/\log m}  \Pr(Z(s)\in B(x))\end{align*}\end{linenomath*}
Use now Corollary~\ref{cor:monotonicity} that gives $p^{Z(m)}(x)\leq \frac{1}{\pi \norm{x}^2}$. Hence, for any $x\in \R^2$ with $\norm{x}>1$, we have \begin{linenomath*}\[\Pr(Z(m)\in B(x))=\int_{B(x)} p^{Z(m)}(y) dy \leq \int_{B(x)} \frac{1}{\pi (\norm{x}-1)^2}dy= \frac{1}{(\norm{x}-1)^2}.\] \end{linenomath*}
Let $m(c)$ be the largest integer $m>0$ such that $cm/\log m\leq 2$. For $m>m(c)$, we have
\begin{linenomath*}\begin{align*}
    \Pr\left (Z(s)\in B(x)\right ) & \leq \max_{s\leq m} \frac{1}{(cm\log m-1)^2}=\frac{1}{(cm\log m-1)^2} 
\end{align*}\end{linenomath*}
Overall, we find that, for $m>m(c)$
\begin{linenomath*}\[  \Pr\left (Z(m)\in B(0)\right ) \leq \frac{1}{(cm/\log m-1)^2} + \frac{c'}{m^2}, \] \end{linenomath*}
which we can bound by $\frac{c_2\log^2 m}{m^2}$ for some constant $c_2>0$. Since $m(c)$ is a constant, there is some other constant $c_3>0$ for which, for any $m\in [2,m(c)]$, we have $\Pr(Z(m)\in B(0))\leq \frac{c_3\log^2 m}{m^2}$. We then obtain, for any $m\geq 2$,
\begin{linenomath*}\[  \Pr(Z(m)\in B(0)) \leq  \frac{\max\{c_2,c_3\} \log^2 m}{m^2}, \] \end{linenomath*}
which concludes the proof of Lemma \ref{LemUB}.
\end{proof}








\end{document}